\documentclass{article}
\usepackage{amssymb,amsthm,amsmath}

\usepackage{graphbox}
\usepackage{xcolor}
\usepackage{grffile}

\def\Me {\mathcal M}

\def\Ha {\mathcal H}
\def\Be {\mathcal B}
\def\Ee {\mathcal E}
\def\Ue {\mathcal U}

\def\Ce {\mathcal C}

\def\Le {\mathcal L}
\def\Fe {\mathsf F}
\def\Ve {\mathcal V}
\def\We {\mathcal W}
\def\states{\mathfrak S}

\def\Mbb{\mathbb M}
\def\Nbb{\mathbb N}
\def\<{\langle}
\def\>{\rangle}
\def\succ{\mathrm{succ}}
\def\pre{\mathrm{pre}}
\def\post{\mathrm{post}}

\def\Tr {{\rm Tr}\,}
\def\ptr {{\rm Tr}}
\newtheorem{thm}{Theorem}
\newtheorem{lemma}{Lemma}
\newtheorem{prop}{Proposition}
\newtheorem{coro}{Corollary}

\theoremstyle{remark}
\newtheorem{ex}{Example}
\newtheorem{rem}{Remark}

\def\llangle{\langle\kern-0.4ex\langle}
\def\rrangle{\rangle\kern-0.4ex\rangle}
\def\doublek{\rrangle\kern-0.3ex\llangle}

\title{A general theory of comparison of quantum channels (and beyond)}
\author{Anna Jen\v cov\'a\thanks{A part of this paper was presented at the Algebraic and Statistical ways into Quantum Resource Theories workshop at BIRS in July 2019}\\ \small \emph{Mathematical Institute, Slovak Academy of Sciences},
\small \emph{\v Stef\'anikova 49, 814 73 Bratislava, Slovakia,} \\ \small \emph{jenca@mat.savba.sk}}
\date{}
\begin{document}

\maketitle
\begin{abstract} We present a general theory of comparison of quantum channels, concerning with the question of simulability or approximate simulability of a given quantum channel by allowed transformations of 
another given channel. 
We introduce a modification of conditional min-entropies, with respect to the set $\Fe$ of allowed
 transformations, and show that  under some conditions on $\Fe$, these quantities
characterize approximate simulability. 
If $\Fe$ is the set of free superchannels in a quantum resource theory of
processes, the modified conditional min-entropies form a complete set of resource monotones.
If the transformations in $\Fe$ consist of  a preprocessing and a postprocessing of specified forms, approximate
simulability is also characterized in terms of success probabilities in certain guessing games, where a preprocessing of
a given form can be chosen and the measurements are restricted. These results are
applied to several specific cases of simulability of quantum channels, including postprocessings, preprocessings
and processing of bipartite channels by LOCC superchannels and by partial superchannels, as well as simulability of sets of
quantum measurements. 

 These questions are first studied in a general setting that is an  extension of the  framework of general probabilistic
theories (GPT), suitable for dealing with channels. Here we prove a general theorem that shows that approximate
simulability can be characterized by comparing outcome probabilities in certain tests. This result is inspired by the
classical Le Cam randomization criterion for statistical experiments and contains its finite dimensional version
 as a special case.

\end{abstract}

\section{Introduction}

For a pair of quantum channels $\Phi_1$ and $\Phi_2$, we consider the following problem: is it possible to simulate 
$\Phi_2$  by transforming $\Phi_1$ by a quantum network of a specified type? 
Since quantum channels are the fundamental objects in quantum
information theory, this question subsumes a variety of special cases already studied extensively in the literature: 
comparison of statistical experiments \cite{matsumoto2010aquantum,buscemi2012comparison}, simulability
of measurements \cite{skrzypczyk2019robustness,guerini2017operational} or more general comparison of channels 
\cite{chefles2009quantum,buscemi2016degradable,buscemi2017comparison,gour2018quantum,gour2019comparison,jencova2016isit,jencova2014randomization,takagi2019general}.
In fact, this kind of questions goes back to the classical theory of comparison of classical statistical experiments 
\cite{blackwell1951comparison,lecam1964sufficiency} (see also \cite{strasser1985mathematical}).
This problem can be also put into the setting of resource theories of processes
\cite{liu2019resource,liu2020operational,gour2019theentanglement,gour2020dynamical} by choosing the
allowed  maps to be the free operations in the theory, whence it becomes the important question of convertibility of one
device to another using free operations.  

Transformations between quantum channels are given by superchannels, consisting of a preprocessing and a postprocessing
channel connected by an ancilla, \cite{chiribella2009theoretical, gour2019comparison}. So the question is the equality or 
approximate equality
 \begin{center}
\begin{minipage}[c]{0.4\textwidth}
\centering
\includegraphics{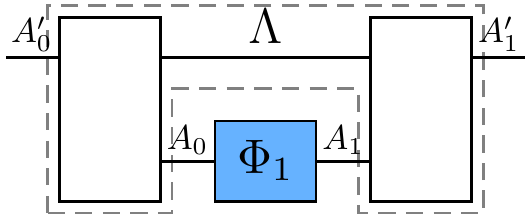}
\end{minipage}
\begin{minipage}[c]{0.01\textwidth}
 $\approx$
\end{minipage}
\begin{minipage}[c]{0.2\textwidth}
\centering
\includegraphics{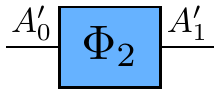}
\end{minipage}
 \end{center}
where $\Lambda$ is a superchannel from a given family.

Two types of characterizations of simulability  are mostly discussed: either by inequalities in (some
modification of) conditional min-entropy (e.g. \cite{gour2018quantum,gour2019comparison}), or in terms of success probabilities in some
discrimination tasks \cite{chefles2009quantum,takagi2019general}. These two
characterizations are closely related, in fact, the latter can be seen as an operational interpretation of the former. 
These conditions provide a complete set of monotones in the given resource theory. 

In  the more general situation where the target channel is simulated only approximately,
there are different possible approaches as to e.g. the distance  measures used to assess the accuracy of the
approximation. A common choice is the diamond norm, which is natural since it is well known as the distinguishability norm for channels
\cite{kitaev1997quantum}. With this choice the problem becomes a direct  extension of the problem of the classical theory of comparison of
statistical experiments. This framework also includes some more specific cases such as quantum dichotomies
\cite{buscemi2019aninformation} (pairs of states), where we can restrict to channels that simulate one of the states
exactly, but the other may differ from the target.

Similarly as in the case of exact simulations, the aim of this paper is to characterize the diamond norm accuracy of the
approximation in two ways: by inequalities in
terms of quantities that can be seen as modifications of the conditional min entropy and by comparison of 
 success probabilities in some types of guessing games. The crucial observation for the first type of characterisation
is the fact that the conditional min entropy is related to a norm: for any state $\rho$ on a composite system $AB$, we
have
\[
H_{\min}(B|A)_\rho=-\log \|\rho\|^\diamond_{B|A}.
\]
Using the Choi isomorphism, the set of operators on the bipartite Hilbert space $\Ha_{AB}$ can be seen as the dual space 
 of the set of linear maps $B(\Ha_A)\to B(\Ha_B)$ and with this identification,  $\|\cdot\|^\diamond_{A|B}$ is the dual norm  to the diamond norm.  A corresponding fact for semifinite von Neumann algebras was also used in
\cite{ganesan2020quantum} to extend the conditional min entropy characterization of the majorization ordering of
bipartite states to the infinite dimensional setting. 

The above duality relation of $H_{\min}$ and the diamond norm is based on  affine duality of convex sets, this was
 observed also in \cite{chiribella2016optimal} and used to extend $H_{\min}$ to quantum networks and their SDP
optimization. While our primary interest lies in simulability of quantum channels, in the first part we will work in 
 a broader setting that in a sense is an extension of  general
probabilistic theories (GPTs), suitable for discussion of various types of quantum channels and
networks. The setting is introduced as a category called $\mathsf{BS}$ and the allowed transformations are specified as a
convex subcategory $\Fe$. The operational theories of \cite{chiribella2011informational} and higher-order theories of
\cite{bisio2019theoretical} can be identified with special objects in $\mathsf{BS}$, in particular, this includes the sets of channels,
superchannels and networks of different types.  Moreover,  the distinguishability norms and the 
 duality  described above are naturally defined here. The proof of the main results, in the GPT setting and for quantum
channels,  is based on the properties of these norms summarized in Proposition \ref{prop:basesec}, together with
 the minimax theorem. 
The advantage of this approach is that it captures the basic mathematical structure that lies behind the result and,
more importantly, it is applicable to a variety of specific cases. We also remark that it can be applied e.g. 
to the case when multiple copies of the channels are studied, under various parallel or sequential schemes.

For two elements $b_1,b_2$ of objects of $\Fe$, the (one-way) $\Fe$-conversion distance $\delta_\Fe(b_1\|b_2)$ is defined as 
the minimal distance 
 we can get to $b_2$ by allowed transformations of $b_1$. The main result in the GPT setting (Theorem \ref{thm:rand}) 
can be interpreted as the fact that $\delta_\Fe$ can be characterized by comparing the outcome probabilities of certain
tests applied to $b_1$ and $b_2$. It is also noted that our setting includes  comparison of classical
statistical experiments (in the simple finite dimensional case), here Theorem \ref{thm:rand} becomes the 
Le Cam randomization criterion \cite{lecam1964sufficiency} of the classical statistical decision theory. The result
is then applied to the case of quantum channels, where  we obtain a characterization of $\delta_\Fe$ in
terms of quantities  that can be interpreted as modifications of the conditional min entropy (Theorem \ref{thm:rand_chans}). In the setting of resource theories, where morphisms in $\Fe$ coincide
with free superchannels, we show that under some mild assumptions these quantities form a complete set of resource
monotones (see e.g. \cite{liu2019resource,liu2020operational, gour2019theentanglement,gour2020dynamical}).

We then turn to the characterization of $\delta_\Fe$  in terms of guessing games. Here we use  a connection between the conditional min entropy and success
probabilities that is specific to the quantum case, so this is studied  only for  quantum channels (apart from 
 examples \ref{ex:experiments} and \ref{ex:measurements} where statistical experiments and measurements are treated 
within the GPT framework). This connection is obtained from an isomorphism between quantum channels and bipartite
measurements which is close to the Choi isomorphism. We assume that the superchannels in $\Fe$ consist of preprocessings and postprocessings
belonging to given sets $\Ce_{\pre}$ and $\Ce_{\post}$. In this case, we find some sufficient conditions (Theorem \ref{thm:rand_chans_psuc}) under which $\delta_\Fe$ is characterized by success
probabilities in guessing games of the following general  form. Given an ensemble  $\Ee=\{\lambda_i, \rho_i\}$ of states $\rho_i$ with prior probabilities
$\lambda_i$, the guessing game is  depicted in the diagram:
\begin{center}
\includegraphics{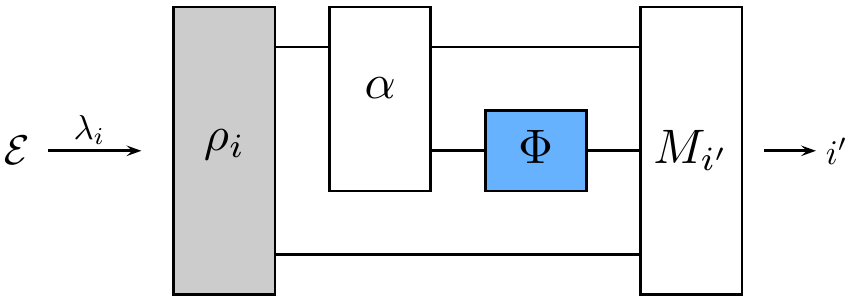} 
\end{center}
Here $\Phi$ is the channel in question ($\Phi_1$ or $\Phi_2$), we may choose the preprocessing $\alpha$
 and the  measurement $M$  from some allowed sets $\Ce_{\pre}$ and $\Me_{\post}$. We also permit  an ancilla between
$\alpha$ and $M$, but this  might be restricted by the allowed sets.
 
These results are applied to some special cases: postprocessing, preprocessing, processing of bipartite channels by
LOCC superchannels and by partial superchannels. We remark that, similarly as LOCC, Theorem \ref{thm:rand_chans_psuc}
can be applied to other cases of restricted resource theories, such as PPT or SEP. In these cases, $\Ce_{\pre}$,
$\Ce_{\post}$ and
$\Me_{\post}$ consist of LOCC (PPT/SEP) channels resp. measurements, see Section \ref{sec:LOCC}. 
 
In the case of postprocessings, we obtain previously known results \cite{jencova2016isit}: $\delta_\Fe$ is characterized
by comparing the  success probabilities of ensembles $(\Phi_i\otimes id)(\Ee)$, see Section \ref{sec:post}.
 For preprocessings, there is only one fixed measurement $M$ in the guessing game but the preprocessings can be chosen
freely.  We also characterize the related preprocessing pseudodistance 
 as a Hausdorff distance of  ranges of the two channels tensored with identity, Section \ref{sec:pre}. The description
of the guessing game for partial processings can be found in Section \ref{sec:partial}.
As another example, we treat classical simulability for sets of quantum measurements and show that it can be formulated
by processing of a certain bipartite channel by a classical-to-classical partial superchannel. We show that
 approximate simulability is characterized by success probabilities, where no ancilla is needed in the guessing game,
 see Section \ref{sec:measurements}.

The outline of the paper is as follows. We start with the general GPT formulation in Section \ref{sec:gpt}. 
Here the category $\mathsf{BS}$ is introduced and the properties of the corresponding norms and their duality are discussed. We then treat the comparison in GPT, the main result is formulated in  Theorem \ref{thm:rand}. 
In Section \ref{sec:comp_channels} we specialize to quantum channels. We first introduce the basic notions and show that
 the sets of channels and superchannels are objects in the category $\mathsf{BS}$. We also recall the connections between 
the diamond norm,  $H_{\min}$ and success probabilities of guessing games (Section \ref{sec:guessing}) that are crucial in further
sections. The main results of this part are Theorem \ref{thm:rand_chans}, characterizing $\delta_\Fe$ in terms of modified
conditional min entropies, and Theorem \ref{thm:rand_chans_psuc}, characterizing $\delta_\Fe$ by success probabilities. 
Applications of these results are contained in Sections \ref{sec:post} - \ref{sec:measurements}.

\section{The GPT formulation}\label{sec:gpt}

General probabilistic theories (GPT) form a framework for description of a large class of physical theories involving
probabilistic processes, see \cite{lami2018nonclassical} for an introduction and background. This framework is built upon basic notions of states and effects and under some
general assumptions on the theories, it can be put into the setting of the theory of (finite dimensional) ordered vector
spaces. The classical and quantum theories are special cases of a GPT (see Example \ref{ex:cl_q} below), which allows to
study some well known quantum phenomena in a broader
context. This is especially useful in the investigation of the mathematical foundations of quantum theory. For us, it
is important that this setting describes the basic mathematical structure underlying the problem of approximate channel
simulability and can be applied in many different situations.

The basic object in GPT is the set of states of a physical system in the theory, represented as a compact convex subset
of an Euclidean space. Such a set can be always seen as a base of a closed convex cone  in a finite dimensional
real vector space. It is clear that the set of { channels}, or physical transformations of the systems in  the
theory, has a convex structure as well and 
 can be, at least formally,  treated as a "state space" in  the convenient framework of  GPT. For example, the set of quantum channels was considered
in this way in \cite{jencova2017conditions}. However, observe that  the set of quantum channels is, by definition, a special subset
of the cone of completely positive maps, but it no longer forms a base of this cone. We will therefore need  a
somewhat more general representation of compact convex sets, described in the next paragraph.

\subsection{Base sections and corresponding norms} \label{sec:base}

Let $\Ve$ be a finite dimensional real vector space and let $\Ve^+\subset \Ve$ be a closed convex cone which is pointed
($\Ve^+\cap (-\Ve^+)=\{0\}$) and generating 
($\Ve=\Ve^+-\Ve^+$). Below, we will say that $\Ve^+$ is a proper cone.  The cone $\Ve^+$ defines a partial order in $\Ve$, defined by $x\le y$ if $y-x\in \Ve^+$. Thus the pair $(\Ve,\Ve^+)$ is an
ordered vector space. 

A convex  subset $B\subset \Ve^+$ is called a  {base section} in $(\Ve,\Ve^+)$ if 
\begin{enumerate}
\item[(i)] $B$ is the base of the cone $\Ve^+\cap \mathrm{span}(B)$;
\item[(ii)] $B\cap \mathrm{int}(\Ve^+)\ne \emptyset$.

\end{enumerate}
 Base sections were studied in
\cite[Appendix]{jencova2014randomization} and in \cite{jencova2014base} in the special case of the space of hermitian
linear operators with the cone of positive operators. In this paragraph, we summarize some of the results.

For the ease of the presentation, it will be convenient to introduce  the category   $\mathsf{BS}$ whose objects are  finite dimensional real vector
spaces $\Ve$ endowed with a fixed proper cone $\Ve^+\subset \Ve$ and a base section $B(\Ve)$ in $(\Ve,\Ve^+)$. 
The morphisms $\Lambda:\Ve\to \We$ in $\mathsf{BS}$ are linear maps  such that $\Lambda(\Ve^+)\subseteq \We^+$ and
$\Lambda(B(\Ve))\subseteq B(\We)$, \cite{plavala2019notes}.

For $\Ve \in \mathsf{BS}$, we define the dual object $\Ve^*\in \mathsf{BS}$ as the dual vector space with the dual cone 
\[
\Ve^{*+}:=\{ \varphi\in \Ve^*, \<\varphi,c\>\ge 0,\ \forall c\in \Ve^+\}
\]
and the dual base section 
\[
B(\Ve^*)=\tilde B(\Ve):=\{\varphi\in \Ve^{*+},\ \<\varphi,b\>=1,\ \forall b\in B(\Ve)\}.
\] 
Note that indeed, $\Ve^{*+}$ is a proper cone in $\Ve^*$ and  $B(\Ve^*)$  is a  base section  in $(\Ve^*,\Ve^{*+})$.
Moreover,  $\Ve^{**}\simeq \Ve$ in $\mathsf {BS}$, where the isomorphism  is given by
the  natural vector space isomorphism $\Ve\to\Ve^{**}$.

 Let us denote 
\begin{align*}
[-B(\Ve),B(\Ve)]:
&=\{x\in \Ve, \ \exists b\in  B(\Ve), -b\le x \le b\}\\
&=\{c_1-c_2,\ c_1,c_2\in \Ve^+,\ c_1+c_2\in B(\Ve)\}. 
\end{align*}
Then we have 
\begin{equation}\label{eq:interval}
\Ve^+\cap [-B(\Ve),B(\Ve)]=[0,B(\Ve)]:=\{x\in \Ve^+,\ \exists b\in B(\Ve),\ x\le b\}.
\end{equation}
 We now define a norm  in $\Ve$ as
\[
\|x\|_{\Ve}:=\max_{\psi\in [-B(\Ve^*),  B(\Ve^*)]} \<\psi,x\>.
\]
The following proposition summarizes some  properties of these norms. Note that since $B(\Ve)$ and $B(\Ve^*)$ are given
by positivity and linear constraints, part (iii) below can be formulated as the primal and dual conic program for computing
the norm $\|\cdot\|_\Ve$.

\begin{prop}\label{prop:basesec} Let $\Ve\in \mathsf{BS}$.  Then 
\begin{enumerate}
\item[(i)] The unit ball of $\|\cdot\|_\Ve$ is $[-B(\Ve),B(\Ve)]$, so that 
\[
\|x\|_\Ve=\min_{b\in B(\Ve)}\min\{\lambda >0,\ -\lambda b\le x\le \lambda b\};
\]
\item[(ii)] the norms  $\|\cdot\|_\Ve$ and $\|\cdot\|_{\Ve^*}$ are mutually dual;

 \item[(iii)] if $c\in \Ve^+$, then 
\[
\|c\|_\Ve=\max_{\varphi\in B(\Ve^*)} \<\varphi,c\>=\min_{b\in B(\Ve)}\min\{\lambda >0,\ c\le \lambda b\};
\]
\item[(iv)] if $b_1,b_2\in B(\Ve)$, then 
\[
\frac12\|b_1-b_2\|_\Ve=\max_{\varphi\in [0,B(\Ve)]} \<\varphi,b_1-b_2\>.
\]
\item[(v)] Let $\Lambda:\Ve\to \We$ be a morphism in $\mathsf{BS}$. Then $\Lambda$ is a contraction with respect to the
corresponding norms:
\[
\|\Lambda(x)\|_\We\le \|x\|_\Ve,\qquad \forall x\in \Ve.
\]

\end{enumerate}

\end{prop}

We next discuss some basic examples of objects in $\mathsf{BS}$.

\begin{ex}[GPT state spaces]\label{ex:GPT} 
 The state spaces of GPT can be also seen as objects in $\mathsf{BS}$. Indeed, let $K$ be any compact convex
subset in $\mathbb R^N$, then there is a finite dimensional real vector space $\Ve$ with a proper cone 
 $\Ve^+$, such that $K$ is a base of $\Ve^+$. This means that there is some functional $u\in
int((\Ve^+)^*)$ such
that
\[
K=\{c\in \Ve^+,\ \<u,c\>=1\}.
\]
The space $\Ve$ with $\Ve^+$ and  $B(\Ve)=K$ is clearly an object in  $\mathsf{BS}$. For the dual object, we have
$B(\Ve^*)=\{u\}$. The norm $\|\cdot\|_\Ve$ is the base norm with respect to $K$ and $\|\cdot\|_{\Ve^*}$ is the order unit norm with
respect to $u$. In this way, the category $\mathsf{BS}$ is a common generalization of order unit and  base normed
spaces.
\end{ex}

\begin{ex}[Classical and quantum state spaces]\label{ex:cl_q}
The prototypical examples in GPT are the classical state space of probability distributions over a finite set, and the
quantum state space of all density density operators on a finite dimensional Hilbert space. In the classical GPT, we have $\Ve=\mathbb R^n$, with the simplicial cone $\Ve^+=(\mathbb R^{+})^n$ and
$K=\{(p_1,\dots,p_n)\in \Ve^{+},\ \sum_ip_i=1\}$ is the probability simplex. The morphisms in $\mathsf{BS}$ between such spaces are precisely the stochastic maps,  given by stochastic matrices.
The norm $\|\cdot\|_\Ve=\|\cdot\|_1$ is the $L_1$-norm. The dual ordered vector space is affinely isomorphic to
$(\Ve,\Ve^+)$ and the unit functional $u=1_n:=(1,\dots,1)$. The dual norm is the $L_\infty$-norm
$\|\cdot\|_{\Ve^*}=\|\cdot\|_\infty$.
In the quantum  case,  $(\Ve,\Ve^+)$ is the space of hermitian operators on a finite dimensional Hilbert space with the cone of positive
operators and the  morphisms are positive trace preserving maps. The dual ordered vector space is again affinely
isomorphic to $(\Ve,\Ve^+)$ and the unit functional is the trace, $u=\Tr$.   
The two dual norms are the trace norm and the operator norm, respectively.
\end{ex}

\begin{ex}[Quantum channels]\label{ex:channels}
The prototypical example of an object in $\mathsf{BS}$ is the set of quantum channels which  is a base section in the
vector space of hermitian maps with the cone of completely positive maps. More details will be given in Section
\ref{sec:diamond}.

\end{ex}

\begin{ex}[$\mathsf{BS}$ morphisms]\label{ex:bsmorph} Let $\Ve$ and $\We$ be  two objects in $\mathsf{BS}$
and let $\Le=\Le(\Ve,\We)$ be the space of linear maps $\Ve\to \We$. With $\Le^+$ the cone of positive maps and 
$B(\Le)=\mathsf{BS}(\Ve,\We)$ the set of all morphisms $\Ve\to \We$, it is easily seen that $\Le$ is again an object in 
 $\mathsf{BS}$. It is this property that makes this category especially useful for description of channels and higher
order theories.

\end{ex}

In general, an object in $\mathsf{BS}$ can be interpreted as a set of special states or devices of a physical theory, so
 we can consider the problem of testing or performing measurements over such sets. A test (or yes-no measurement)
is identified with an affine map  $B(\Ve)\to [0,1]$, assigning to each element the probability of the "yes" outcome.
Such maps correspond to elements $\varphi\in [0,B(\Ve^*)]$, similarly to the GPT setting, these elements will be called
effects. Any measurement with $k$ outcomes is given by a $k$-tuple of effects $\{\varphi_i\}_{i=1}^k$ such that 
 $\sum_i \varphi_i\in B(\Ve^*)$. Note that in the case of quantum channels this corresponds precisely to the quantum
testers of \cite{chiribella2009theoretical}.

It is not difficult to see that the restriction of $\|\cdot\|_\Ve$ to  $\mathrm{span}(B(\Ve))$ coincides with
the base norm with respect to the base $B(\Ve)$ of the cone $\Ve^+\cap \mathrm{span}(B(\Ve))$. The advantage of our
extended definition is that the dual space is now an object of the same  category. As the base norm,  $\|\cdot\|_\Ve$
has an operational interpretation as a distinguishability norm for elements in $B(\Ve)$. It can be seen that this
 holds also if the discrimination procedures are given by the effects as above, indeed, using Proposition
\ref{prop:basesec} (iv) we see that that the optimal probability of correctly distinguishing two elements  
$b, b'\in B(\Ve)$ each of which has a prior probability 1/2 is $\frac12(1+\frac12\|b-b'\|_\Ve)$.

More generally, an ensemble on $\Ve$  is a finite sequence $\Ee=\{\lambda_i,x_i\}_{i=1}^k$ of elements  $x_i\in B(\Ve)$ and
prior probabilities $\lambda_i$. The interpretation is that   $x_i$ is prepared  with probability
$\lambda_i$ and the task is to guess which element was prepared.
Any guessing procedure is described by a $k$-outcome measurement $\psi=\{\psi_i\}$, the
value $\<\psi_i,x_j\>$ is interpreted as the probability of guessing $i$ if the true state was $x_j$.
  The average success probability using $\psi=(\psi_i)$ is 
\[
P_{\mathrm{succ}}(\Ee,\psi):=\sum_i\lambda_i\<\psi_i,x_i\>
\]
 and the optimal success probability for $\Ee$ is
\[
P_{\succ}(\Ee):= \max_{\psi} P_{\succ}(\Ee,\psi).
\]

\subsection{Comparison in GPT}\label{sec:comp_gpt}

We now formulate the comparison problem in the above setting and prove a general theorem which in later
sections will be applied to quantum channels. Assume that a subcategory $\Fe$ in $\mathsf{BS}$ is given, such that for any $\Ve,\We\in \Fe$, the set $\Fe(\Ve,\We)$ of all morphisms
$\Ve\to \We$ in $\Fe$ is convex (we will say in this case that $\Fe$ is a convex subcategory).  For two objects $\Ve_1,\Ve_2\in \Fe$, let $b_1\in B(\Ve_1)$,
$b_2\in B(\Ve_2)$. The (one-way) $\Fe$-conversion distance $\delta_\Fe(b_1\|b_2)$ is defined as the minimum distance we can get to $b_2$ by images of $b_1$ under all morphisms $\Ve_1\to \Ve_2$ in $\Fe$: 
\[
\delta_\Fe(b_1\|b_2):=\inf_{\Lambda\in \Fe(\Ve_1,\Ve_2)} \|\Lambda(b_1)-b_2\|_{\Ve_2}.
\]
We also define the {$\Fe$-distance} of $b_1$ and $b_2$ as
\[
\Delta_\Fe(b_1,b_2):=\max\{\delta_\Fe(b_1\|b_2),\delta_\Fe(b_2\|b_1)\}.
\]
\begin{prop} $\Delta_\Fe$ is a pseudometric on the set $\{b\in B(\Ve),\ \Ve\in \Fe\}$.

\end{prop}

\begin{proof}  The only thing to prove
is the triangle inequality. So let $\Ve_1,\Ve_2,\Ve_3\in \Fe$ and let $b_i\in \Ve_i$, $i=1,2,3$.
Let
$\Theta\in \Fe(\Ve_2,\Ve_3)$ and let for $\mu>0$, $\Lambda_\mu\in \Fe(\Ve_1,\Ve_2)$
be such that $\delta_\Fe(b_1\|b_2)+\mu\ge \|\Lambda_\mu(b_1)-b_2\|_{\Ve_2}$. Then
\begin{align*}
\delta_\Fe(b_1\|b_3)&\le \|\Theta\circ\Lambda_\mu(b_1)-b_3\|_{\Ve_3}\le
\|\Theta(\Lambda_\mu(b_1)-b_2)\|_{\Ve_3}+\|\Theta(b_2)-b_3\|_{\Ve_3}\\
&\le \mu+\delta_\Fe(b_1\|b_2)+\|\Theta(b_2)-b_3\|_{\Ve_3}.
\end{align*}
since this holds for all $\Theta\in \Fe(\Ve_2,\Ve_3)$ and all $\mu>0$, we get the result.

\end{proof}

The following data processing inequalities for $\delta_\Fe$ follow  easily from Prop. \ref{prop:basesec} (v) (cf.
\cite[Prop. 3]{jencova2016isit}).

\begin{prop} Let $\Ve_1,\Ve_2,\Ve_3\in \Fe$, $b_i\in B(\Ve_i)$, $i=1,2,3$. Then
\begin{enumerate}
\item[(i)] For any $\Lambda\in \Fe(\Ve_1,\Ve_2)$, $\delta_\Fe(\Lambda(b_1)\|b_3)\ge \delta_\Fe(b_1\|b_3)$;
\item[(ii)] For any $\Theta\in \Fe(\Ve_2,\Ve_3)$, $\delta_\Fe(b_1\|\Theta(b_2))\le \delta_\Fe(b_1\|b_2)$;
\end{enumerate}

\end{prop}

We now turn to the main result of this section, contained in Theorem \ref{thm:rand} below. Note that since the elements
$\varphi\in \Ve^{*+}$, $\psi\in \We^{*+}$ in (ii) and (iii) of this theorem can 
always be normalized and the positive part of the unit ball of $\|\cdot\|_{\Ve^*}$ coincides with the set of all tests on 
 $\Be(\Ve)$, this theorem says that the  conversion distance  can be characterized by comparing the probabilities of the 
"yes" outcome for certain tests applied to $b_1$ and $b_2$.

 For the proof of Theorem \ref{thm:rand},
we will need the following minimax theorem (cf. \cite[Thm.48.5]{strasser1985mathematical}).

\begin{thm}[Minimax theorem]\label{thm:minimax} Let $T$ be a convex and compact subset of a locally convex space and let $Y$ be a convex
subset of a vector space. Let $f:T\times Y\to \mathbb R$ be convex in $y$ and continuous and concave in $t$. Then
\[
\inf_{y\in Y}\sup_{t\in T} f(t,y)=\sup_{t\in T}\inf_{y\in Y} f(t,y).
\]

\end{thm}

\begin{thm} \label{thm:rand}  Let $\Ve_1,\Ve_2\in \Fe$ and let $b_1\in B(\Ve_1)$, $b_2\in
B(\Ve_2)$. Let $\epsilon\ge 0$. Then the following are equivalent.
\begin{enumerate}
\item[(i)] $\delta_\Fe(b_1\|b_2)\le \epsilon$;
\item[(ii)] for all $\varphi\in \Ve_2^{*+}$, there is some $\Theta\in \bar{\Fe}(\Ve_1,\Ve_2)$ such that
\[
\<\varphi,b_2\>\le \<\varphi,\Theta(b_1)\>+ \frac{\epsilon}2\|\varphi\|_{\Ve^*},
\]
here $\bar{\Fe}(\Ve_1,\Ve_2)$ is the closure of $\Fe(\Ve_1,\Ve_2)$ in $\mathsf{BS}(\Ve_1,\Ve_2)$;
\item[(iii)] for all $\We\in \Fe$ and all $\psi\in \We^{*+}$, we have
\[
\sup_{\Theta\in \Fe(\Ve_2,\We)} \<\psi,\Theta(b_2)\>\le \sup_{\Theta'\in \Fe(\Ve_1,\We)}\<\psi, \Theta'(b_1)\>+ 
\frac{\epsilon}2\|\psi\|_{\We^*}.
\]
\end{enumerate}

\end{thm}

\begin{proof} Let $\epsilon'>\epsilon$ and let  $\Lambda_0\in \Fe(\Ve_1,\Ve_2)$ be such that $\|\Lambda_0(b_1)-b_2\|_{\Ve_2}\le
\epsilon'$.  Let $\We\in \Fe$ and let  $\psi\in \We^{*+}$. For any $\Theta\in \Fe(\Ve_2,\We)$, we have
\begin{align*}
\<\psi,\Theta(b_2)\>&=\<\psi,\Theta\circ\Lambda_0(b_1)\>+\<\psi, \Theta(b_2-\Lambda_0(b_1))\>\\
&\le \sup_{\Theta'\in \Fe(\Ve_1,\We)}\<\psi,\Theta'(b_1)\>+\frac12\| \Theta(b_2-\Lambda_0(b_1))\|_{\We}\|\psi\|_{\Ve^*}\\
&\le \sup_{\Theta'\in \Fe(\Ve_1,\We)}\<\psi,\Theta'(b_1)\>+\frac{\epsilon'}2\|\psi\|_{\Ve^*},
\end{align*}
where we have used Prop. \ref{prop:basesec} (iv) and (v). Since this holds for all $\Theta\in \Fe(\Ve_2,\We)$ and
$\epsilon'>\epsilon$, this proves that (i) implies (iii). Since (iii) obviously implies (ii), it is enough to prove (ii)
$\implies$ (i).

 By Prop. \ref{prop:basesec}  (iv), we have
\[
\frac12\delta_\Fe(b_1\|b_2)=\inf_{\Lambda\in \Fe(\Ve_1,\Ve_2)}\sup_{\varphi\in [0,B(\Ve^*_2)]}\<\varphi,b_2-\Lambda(b_1)\>.  
\]
Note that by  Prop. \ref{prop:basesec} (i), 
the set  $[0,B(\Ve_2^*)]=\Ve_2^{*+}\cap [-B(\Ve_2^*),B(\Ve^*_2)]$ is convex and  compact  and the map 
$(\varphi,\Lambda)\mapsto \<\varphi,b_2-\Lambda(b_1)\>$ is linear in both components, 
so that we may apply the minimax theorem (Thm. \ref{thm:minimax}). Assume that (ii) holds, then 
\begin{align*}
\frac12\delta_\Fe(b_1\|b_2) &= \sup_{\varphi\in [0,B(\Ve^*_2)]}\inf_{\Lambda\in
\Fe(\Ve_1,\Ve_2)}\<\varphi,b_2-\Lambda(b_1)\>\\
&= \sup_{\varphi\in [0,B(\Ve^*_2)]}\left( \<\varphi,b_2\>- \sup_{\Lambda\in \Fe(\Ve_1,\Ve_2)}
\<\varphi,\Lambda(b_1)\>\right)\le \frac{\epsilon}2.
\end{align*}

\end{proof}

We next some examples of an application of Theorem \ref{thm:rand}, the first two of which show the relation to 
the theory of comparison of statistical experiments. In these examples, $\Ve\in \mathsf{BS}$ is
such that $K=B(\Ve)$ is a base of $\Ve^+$. The dual object $\Ve^*$ is the order unit space $(\Ve^*,\Ve^{*+},u)$, where
$u\in \Ve^{*+}$ is the functional determined by $\<u,x\>=1$ for all $x\in K$.

\begin{ex}[Statistical experiments] \label{ex:experiments}  A (finite) statistical experiment in $\Ve$ is a finite set $x_1,\dots,x_k\in K$.
The set of all experiments with fixed  $k$ and $\Ve$ is an object in $\mathsf{BS}$. Indeed, let $\Ve^k=\oplus_{i=1}^k \Ve$ and
$\Ve^{k+}=\oplus_{i=1}^k \Ve^+$, then $K^k=\oplus_{i=1}^k K$ is a base section in 
the ordered vector space $(\Ve^k, \Ve^{k+})$. 
As for the dual object, $\Ve^{k*}=\oplus_{i=1}^k \Ve^*$ and 
$\Ve^{k*+}=\oplus_{i=1}^k \Ve^{*+}$. Moreover, it is easily checked that the dual section is
\[
\Be(\Ve^{k*})=\tilde K^k=\{(p_1u,\dots,p_ku),\ p_i\in [0,1], \ \sum_ip_i=1\}
\]
and 
\begin{align*}
\|v\|_{\Ve^k}&=\max_i \|v_i\|_\Ve, \qquad v=(v_i)\in \Ve^k\\
\|\psi\|_{\Ve^{k*}}&=\sum_i \|\psi_i\|_{\Ve^*},\qquad \psi=(\psi_i)\in \Ve^{k*}.
\end{align*}
Let $\mathsf F$ be the subcategory whose objects are statistical experiments with fixed $k$ and morphisms in $\mathsf F(\Ve^k,\mathcal W^k)$ are given 
by $\Phi^k$ with $\Phi\in \mathcal F(\Ve,\We)\subseteq \mathsf{BS}(\Ve,\mathcal W)$ for some convex subset $\mathcal
F(\Ve,\We)$. 

Let $x=(x_i)\in B(\Ve^k)$, 
$y=(y_i)\in B(\mathcal W^k)$, then the $\Fe$-conversion distance has the form 
\[
\delta_{\mathsf F}(x\|y)=\inf_{\Phi\in \mathcal F} \max_i\|\Phi(x_i)-y_i\|_\Ve.
\]
Under an obvious normalization, Theorem \ref{thm:rand} (iii) says that $\delta_{\mathsf F}(x\|y)\le \epsilon$ if and only if for any $(\psi_i)\in
 \Ve^{k*+}$ with $\sum_i\|\psi_i\|_{\Ve^*}\le 1$, we have
\[
\sup_{\Phi\in \mathcal F}\sum_i \<\psi_i,\Phi(y_i)\>\le \sup_{\Phi'\in \mathcal F}\sum_i
\<\psi_i,\Phi'(y_i)\>+\frac{\epsilon}2.
\]
Since the norm $\|\cdot\|_{\Ve^*}$ is the order unit norm with respect to $u$, any $(\psi_i)$ of the above form
satisfies $0\le \psi_i\le \|\psi_i\|_{\Ve^*}u$  and  $\sum_i\psi_i\le \sum_i\|\psi_i\|_{\Ve^*}u\le u$. Adding a positive
element $\psi_{k+1}=u-\sum_i \psi_i$, the collection $(\psi_i)_{i=1}^{k+1}$ becomes a measurement with $k+1$ outcomes
and the value $\frac1k\sum_i \<\psi_i,\Phi(y_i)\>=: \tilde P_{\succ}(\Ee,(\psi_i))$ becomes the success probability for
the ensemble $\Ee=\{\frac1k,\Phi(y_i)\}$ in the {inconclusive} discrimination with the measurement $(\psi_i)$, here
$\psi_{k+1}$ represents the inconclusive outcome. Now we see that $\delta_{\mathsf F}(x\|y)\le \epsilon$ if and only if
 for any measurement $(\psi_i)_{i=1}^{k+1}$,
\[
\sup_{\Phi\in \mathcal F}\tilde P_{\succ}(\{\frac1k,\Phi(y_i)\},(\psi_i))\le\sup_{\Phi'\in \mathcal F}\tilde
P_{\succ}(\{\frac1k,\Phi'(x_i)\},(\psi_i))+\frac{\epsilon}{2k},  
\]
extending the result \cite[Cor. 15]{takagi2019general}.

\end{ex}

\begin{ex}[Le Cam randomization criterion for classical statistical experiments]\label{ex:lecam} In the setting of the 
previous example, let us further restrict $\mathsf F$ to $\Ve^k$ for classical state spaces $\Ve$ (see Example
\ref{ex:cl_q}) and let $\mathcal F$ be 
the set of all stochastic maps. Here the objects of $\Fe$ are sets of classical statistical experiments with
$k$ elements. If $(p^i)$ is such an experiment and $q^i=\Phi(p^i)$ for some stochastic map $\Phi$, we say that $(q^i)$
is a randomization of $(p^i)$, so that $\Fe$ is the category of (finite dimensional) classical statistical experiments with
randomizations. Moreover, $\delta_\Fe((p^i)\|(q^i))$ is the Le Cam deficiency of $(p^i)$ with respect to $(q^i)$ and 
 $\Delta_\Fe$ becomes the the Le Cam distance \cite{lecam1964sufficiency}.  

The basic idea of comparison of statistical experiments that goes back to Blackwell \cite{blackwell1951comparison}
is to compare classical statistical experiments
by the performance of decision rules. Here the task is to choose a decision $d$ from a finite set $\{1,\dots, D\}$ 
using data that is known to be drawn according to one of the distributions in  $\{p^1,\dots,p^k\}\subset \Delta_m$.
The decision rules are given by stochastic maps $\Theta:\Delta_m\to \Delta_D$, where $\Theta(p)(d)$ is the probability
that $d$ is chosen if the data was sampled from the distribution  $p\in\Delta_m$. 
To each pair $i=1,\dots,k$ and $d=1,\dots,D$ a value $g(i,d)\ge 0$ is assigned expressing the gain obtained if $d$ was
chosen while the true distribution was $p^i$. Under a prior distribution $\lambda=(\lambda_1,\dots,\lambda_k)$,
the average gain of the decision rule $\Theta$ is given by
\[
G(\{p^i\}, \lambda, g, \Theta)=\sum_i\sum_d\lambda_i \Theta(p^i)(d)g(i,d)=\<\psi, \Theta(\{p^i\})\>,
\]   
where $\psi\in (\mathrm R_+^D)^k$ is given by $\psi^i_{d}=\lambda_ig(i,d)$, note also that we have 
$\|\psi\|_{\Ve^{k*}}=\sum_i\lambda_i\max_dg(i,d)$. The celebrated Le Cam randomization criterion
\cite{lecam1964sufficiency} says that the
deficiency of the experiment $(p^i)$ with respect to $(q^i)$ can be obtained by comparing 
 the optimal values of $G$ achievable by the two 
experiments, more precisely that $\delta_\Fe((p^i)\|(q^i))\le \epsilon$
if and only if for all gain functions $g$ and prior distributions
$\lambda$, we have
\[
\sup_{\Theta}G(\{q^i\},\lambda,g,\Theta)\le
\sup_{\Theta'}G(\{p^i\},\lambda,g,\Theta')+\frac{\epsilon}2\sum_i\lambda_i\max_d g(i,d). 
\]
In the finite dimensional setting, this is precisely the statement of Theorem \ref{thm:rand}. 
Therefore, Theorem \ref{thm:rand} can be seen as the most general GPT form of the Le Cam randomization theorem.
 
Similarly, restricting the objects $\Ve$ to quantum state spaces and letting $\mathcal
F$ be the set of all quantum channels, we obtain the quantum version of the randomization criterion, cf.
\cite{matsumoto2010aquantum}.

\end{ex}

\begin{ex}[Measurements]  \label{ex:measurements} A measurement (with $k$ outcomes) is a collection $M=(M_i)\in \Ve^{k*+}$, $\sum_i M_i=u$.
It easy to see that the set $\mathcal M_k(\Ve)$ of all such  measurements  
is a base section in $(\Ve^{k*}, \Ve^{k*+})$ we will denote this object of $\mathsf{BS}$  
by $\Ve^k_{\mathcal M}$.
The dual object  is $\Ve^k,\Ve^{k+}$ with the base section
\[
\tilde {\mathcal M}_k(\Ve)=\{(x,\dots,x),\ x\in K\}
\]
Note that for  $v\in \Ve^{k+}$, the dual norm is
\[
\|v\|_{\Ve^{k*}_{\mathcal M}}= \max_{M\in \mathcal M_k(\Ve)} \sum_i \<M_i,v_i\>.
\]
Dividing $v$ by $c:=\sum_i \<u,v_i\>$ and noting that $c^{-1}v_i= \lambda_ix_i$ with $x_i\in K$
and some probabilities $\lambda_i$, we obtain an ensemble: $\mathcal E:=\{\lambda_i,x_i\}$ such that
\[
\|v\|_{\Ve^{k*}_{\mathcal M}}=cP_{\succ}(\mathcal E).
\]
Again, let $\Fe$ be the subcategory with objects $\Ve_\Me^k$ (with $k$ fixed) and  morphisms $\Phi^{*k}:\Ve_\Me^k\to
\mathcal W_\Me^k$ 
for $\Phi\in \mathcal F(\We,\Ve)\subseteq \mathsf{BS}(\mathcal W,\Ve)$.
The Theorem \ref{thm:rand} tells us that $\delta_{\mathsf F}(M\|N)\le \epsilon$ if and only if, for any ensemble $\mathcal E$ on $\Ve$,
\[
P_{\succ}(\mathcal E,N)\le \sup_{\Phi\in \mathcal F} P_{\succ}(\Phi(\mathcal E),N)+\frac{\epsilon}2 P_{\succ}(\mathcal E),
\]
extending the result of \cite[Thm. 14]{takagi2019general}.

\end{ex}

More examples will be treated in the next section.

\section{Comparison of quantum channels}\label{sec:comp_channels}

In this section we will present the sets of quantum channels and superchannels as objects in $\mathsf{BS}$ and 
show how Theorem \ref{thm:rand}  applies, under some conditions  on the subcategory $\Fe$. 
We need some preparation first.

\subsection{Basic ingredients}

%citations overall
 In what follows, $\Ha_A,\Ha_B,\dots$ will always denote a finite dimensional Hilbert space, labelled by the system it represents. The Hilbert space will often be referred to by its label, so we denote by 
$\Be(A)$ the set of bounded operators on $\Ha_A$, similarly, $\Be_h(A)$ denotes the set of self-adjoint operators,
 $\Be_+(A)$ the set of positive operators  and $\states(A)$ the set of states on $\Ha_A$. We will also put
 $d_A:=\dim(\Ha_A)$ and $I_A$ denotes the  identity operator on $\Ha_A$. The trivial Hilbert space $\mathbb C$ will be labeled by 1.

For $W\in \Be(A_0)$, we will use
the notation (cf. \cite{chiribella2009theoretical})
\begin{equation}\label{eq:dobleket}
|W\rrangle:= \sum_i W|i\>_{A_0}\otimes |i\>_{A_0}=\sum_i|i\>_{A_0}\otimes W^{\mathsf{T}}|i\>_{A_0},
\end{equation}
here $W^{\mathsf{T}}$ denotes the transpose of $W$ in the standard basis $\{|i\>\}$.

\subsubsection{Linear maps  and Choi representation}
Let $\Le(A_0,A_1)$ denote the set of hermitian linear maps $\Be(A_0)\to \Be(A_1)$, that is, linear maps satisfying
\[
\Phi(X^*)=\Phi(X)^*,\qquad X\in \Be(A_0).
\] 
Let $\Le_+(A_0, A_1)$ denote the subset of completely positive maps in $\Le(A_0,A_1)$ and $\Ce(A_0,A_1)$ the set of
quantum channels, that is, trace preserving maps in $\Le_+(A_0,A_1)$.

The Choi matrix of $\Phi\in \Le(A_0, A_1)$ is defined as $C_\Phi:= (\Phi\otimes id_{A_0})(|I_{A_0}\doublek I_{A_0}|)$. 
The map $\Phi\mapsto C_\Phi$  establishes a linear isomorphism between  $\Le(A_0,A_1)$ and $\Be_h(A_1A_0)$ that maps 
$\Le_+(A_0,A_1)$ onto $\Be_+(A_1A_0)$ and  $\Ce(A_0, A_1)$ onto the set
\[
\{X\in \Be_+(A_1A_0),\ \ptr_{A_1}[X]=I_{A_0}\}.
\]

\subsubsection{Diagrams}

We will make use of the common diagrammatic representation of  maps in $\Le(A_0, A_1)$ as 
\begin{center}
\includegraphics{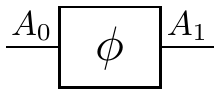} 
\end{center}
If some of the systems is trivial, the corresponding wire will be omitted.
The special symbols
\begin{center}
\includegraphics{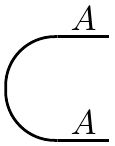} \hspace{1cm} \includegraphics{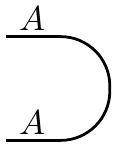}
\end{center}
will represent $|I_{A}\doublek I_{A}|$ as  a preparation (a map $1\to AA$)  and as an effect (a map
$AA\to 1$)  respectively. 
In this way, we may write the Choi isomorphism and its inverse as
\begin{center}
\begin{minipage}[c]{0.2\textwidth}
\centering
\includegraphics{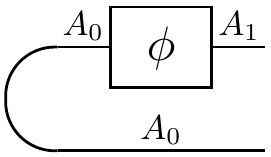}
\end{minipage}
\begin{minipage}[c]{0.01\textwidth}
 = 
\end{minipage}
\begin{minipage}[c]{0.2\textwidth}
\centering
\includegraphics{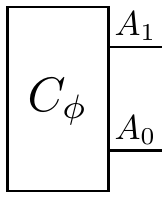}
\end{minipage}
\hspace{20pt}
\begin{minipage}[c]{0.2\textwidth}
\centering
\includegraphics{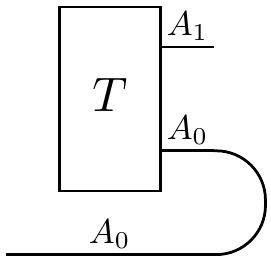}
\end{minipage}
\begin{minipage}[c]{0.01\textwidth}
 = 
\end{minipage}
\begin{minipage}[c]{0.2\textwidth}
\centering
\includegraphics{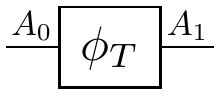}
\end{minipage}

\end{center}
We will use similar symbols for the maximally entangled state $\psi^{A}:=d_{A}^{-1}|I_A\doublek I_A|$:
\begin{center}
\includegraphics{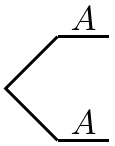} \hspace{1cm} \includegraphics{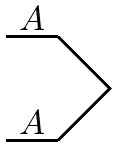}
\end{center}

\subsubsection{The link product}
The Choi matrix of a composition of maps is given by the link product of the respective Choi matrices,
\cite{chiribella2009theoretical}. 
For general multipartite matrices $X\in \Be(AB)$ and $Y\in \Be(BC)$, the link product is defined as
\[
X*Y=\ptr_{B} [(X\otimes I_C)(I_A\otimes Y^{\mathsf{T}_B})]
\]
here $(\cdot)^{\mathsf{T}_B}$ denotes the partial transpose on the system $B$. Diagrammatically:
\begin{center}
\begin{minipage}[c]{0.2\textwidth}
\centering
\includegraphics{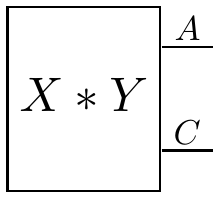}
\end{minipage}
\begin{minipage}[c]{0.01\textwidth}
 = 
\end{minipage}
\begin{minipage}[c]{0.2\textwidth}
\centering
\includegraphics{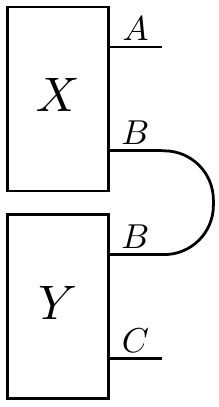}
\end{minipage}
\end{center}

The link product is commutative (up to the order of the
spaces) and associative provided that the three matrices have no labels in common. The order of the spaces is not taken
into account, applying an appropriate unitary conjugation swapping the spaces in the tensor products if necessary, so,
for example, if $X\in \Be(AB_1B_2)$ and $Y\in \Be(B_2B_1C)$, then 
\begin{equation}\label{eq:swaplink}
X*Y\equiv X*\mathcal U_{B_1,B_2}(Y),
\end{equation}
where $\mathcal U_{B_1,B_2}$ is the conjugation by the unitary swap $ U_{B_1,B_2}:\Ha_{B_1B_2}\to \Ha_{B_2B_1}$.

\subsubsection{Superchannels and 2-combs}
A quantum superchannel is a special type of causal quantum network that transforms channels into channels,
with possibly different input and output systems. Any superchannel $\Lambda$ that maps $\Ce(A_0,A_1)$ into
$\Ce(A'_0,A'_1)$ is a channel in $\Ce(A'_0A_1, A'_1A_0)$, consisting of a pre-processing channel $\Lambda_{\pre}\in
\Ce(A'_0,RA_0)$ and a post-processing channel $\Lambda_{\post}\in \Ce(RA_1,A'_1)$, where $R$ is some ancilla
\cite{chiribella2009theoretical}.  We will write
$\Lambda=\Lambda_{\pre}*\Lambda_{\post}$ for this concatenation of channels. The set of all such superchannels will be
denoted by $\Ce_2(A,A')$, where we used the abbreviation $A=A_0A_1$, $A'=A'_0A'_1$. Diagrammatically, $\Lambda$ can be represented
as
\begin{center}
\begin{minipage}[c]{0.4\textwidth}
\centering
\includegraphics{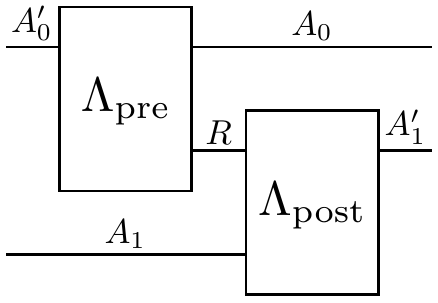}
\end{minipage}
\begin{minipage}[c]{0.05\textwidth}
 = 
\end{minipage}
\begin{minipage}[c]{0.4\textwidth}
\centering
\includegraphics{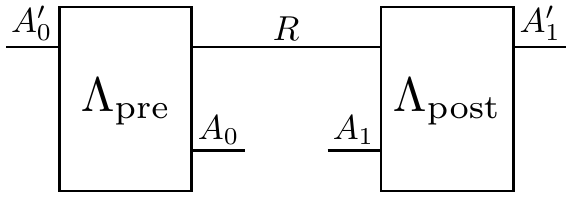}
\end{minipage}
\end{center}
and acts on a map $\phi$ as  
\begin{center}
\includegraphics{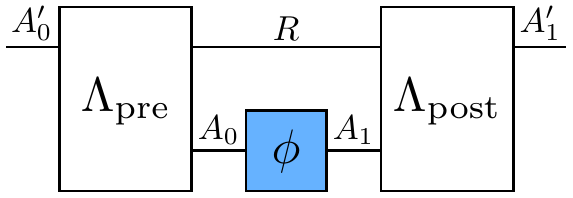}
\end{center}

The Choi matrices of superchannels are called 2-combs in \cite{chiribella2009theoretical}. Using the link product and
its properties, we have $C_\Lambda=C_{\Lambda_{\pre}}*C_{\Lambda_{\post}}$ and 
$C_{\Lambda(\Phi)}=C_\Lambda*C_\Phi=C_{\Lambda_{\pre}}*C_\Phi*C_{\Lambda_{\post}}$. An element $C\in \Be_+(A'_1A_0A'_0A_1)$
is a 2-comb if and only if
\begin{equation}\label{eq:comb}
\ptr_{A'_1}[C]= I_{A_1}\otimes C_2,\qquad \ptr_{A_0}[C_2]=I_{A'_0},
\end{equation}
which means that  $C_2$ is the Choi matrix of some channel in $\Ce(A'_0,A_0)$.

\subsubsection{Diamond norm and the conditional min-entropy} \label{sec:diamond}

It is not difficult to see that $\Le(A_0,A_1)$ with the cone $\Le_+(A_0,A_1)$ and 
$B(\Le(A_0,A_1))=\Ce(A_0,A_1)$ is an object in $\mathsf{BS}$, similarly for the set of superchannels. It was observed in
\cite{jencova2014base} that
in these cases the structures described in Section \ref{sec:base}
yield some well known quantities. This will be discussed in the present and the next section, see \cite{jencova2014base}
for more details.

We will use the identification of the dual space $\Le^*(A_0,A_1)$ with  $\Be_h(A_0A_1)$, with duality for
$X\in\Be_h(A_0A_1)$ and $\phi\in \Le(A_0,A_1)$ given by
\begin{align}
\<X,\phi\>&:= \llangle I_{A_1}|
(\phi\otimes id)(X)|I_{A_1}\rrangle = \Tr[XC_{\phi^*}] = \Tr[XC_{\phi^{\mathsf{T}}}^{\mathsf{T}}]  \notag\\
&=X*C_{\phi^{\mathsf{T}}}=X*C_\phi. \label{eq:dualitylink}
\end{align}
Diagrammatically, this can be expressed as
\begin{center}
\begin{minipage}[c]{0.1\textwidth}
\centering
$\<X,\phi\> =$
\end{minipage}
\begin{minipage}[c]{0.3\textwidth}
\centering
\includegraphics{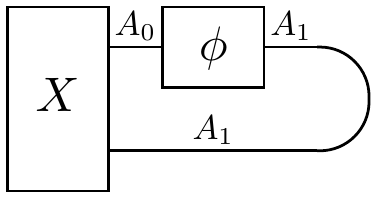}
\end{minipage}
\begin{minipage}[c]{0.01\textwidth}
 = 
\end{minipage}
\begin{minipage}[c]{0.3\textwidth}
\centering
\includegraphics{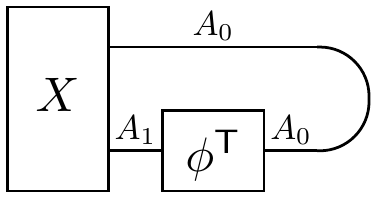}
\end{minipage}

\end{center}

Here  $\phi^*, \phi^{\mathsf{T}} \in \Le(A_1,A_0)$ are maps determined by
\begin{align*}
\Tr[X\phi^*(Y)]=\Tr [\phi(X)Y],\qquad \phi^{\mathsf{T}}(Y)=(\phi^*(Y^{\mathsf{T}}))^{\mathsf{T}}
\end{align*}
for $X\in \Be_h(A_0)$ and $Y\in \Be_h(A_1)$. The last equality in \eqref{eq:dualitylink} follows from
\eqref{eq:swaplink} and 
$C_{\phi^{\mathsf{T}}}=\Ue_{A_1,A_0}(C_\phi)$. Note that by \eqref{eq:dualitylink} we  also have
\begin{equation}\label{eq:dualityswp}  
\<X,\phi\>= \Tr[X\Ue_{A_1,A_0}(C_\phi)^{\mathsf{T}}]=\Tr [\Ue_{A_1,A_0}^*(X^{\mathsf{T}})C_\phi].
\end{equation}

With these identifications, the dual cone is $\Le_+^*(A_0,A_1)\simeq \Be_+(A_0A_1)$ and the dual section
\[
 B(\Le^*(A_0,A_1)) = \tilde\Ce(A_0,A_1)=\{ \sigma_{A_0}\otimes I_{A_1},\ \sigma_{A_0}\in \states(A_0)\}.
\]
The corresponding base section norm is the diamond norm
\[
\|\Phi\|_{\Le(A_0,A_1)}=\|\Phi\|_\diamond:=\max_{\rho\in \states(A_0A_0)} \|(\Phi\otimes id)(\rho)\|_1,\qquad \Phi\in \Le(A_0,A_1),
\] 
well known as the distinguishability norm for quantum channels, \cite{kitaev1997quantum, watrous2018thetheory}.

\begin{rem}\label{rem:dual}
Using the Choi representation, we may also identify  $\Le^*(A_0,A_1)\simeq \Le(A_1,A_0)$, with duality
$\<\cdot,\cdot\>_*$ given as
\[
\<\psi,\phi\>_*:=\<C_\psi,\phi\>=\tau(\phi\circ\psi),
\]
where the functional $\tau:\Le(A_0,A_0)\to \mathbb R$ is given by
\[
\tau(\xi)= \sum_{i,j} \Tr \left[|i\>\<j| \xi(|i\>\<j|)\right]=\llangle I_{A_0}|C_\phi|I_{A_0}\rrangle,
\] 
in diagram
\begin{center}
\begin{minipage}[c]{0.1\textwidth}
\centering
$\tau(\xi) =$
\end{minipage}
\begin{minipage}[c]{0.3\textwidth}
\centering
\includegraphics{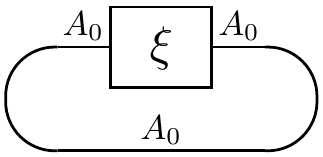}
\end{minipage}
\begin{minipage}[c]{0.01\textwidth}
 = 
\end{minipage}
\begin{minipage}[c]{0.3\textwidth}
\centering
\includegraphics{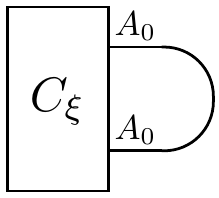}
\end{minipage}
\end{center}
Choosing the Hilbert-Schmidt inner product and the basis $\{|i\>\<j|\}$ in 
 $\Be(A_0)$, we see that $\tau$ is the usual trace of elements in $\Le(A_0,A_0)$ as linear maps. 
In this identification,  the dual section becomes the set of replacement channels in $\Ce(A_1,A_0)$, mapping all states in $\states(A_1)$
to a fixed state $\sigma_{A_0}\in \states(A_0)$. 
%obrazok

\end{rem}

Let us introduce the notation 
\[
\|\cdot\|_{A_1|A_0}^\diamond:=\|\cdot\|_{ \Le^*(A_0,A_1)}
\]
for the dual norm.  By Prop. \ref{prop:basesec} (iii), we have for ${\rho}\in
\Be_+(A_0A_1)$:
\begin{equation}\label{eq:cond_min}
\|{\rho}\|^\diamond_{A_1|A_0}=\min_{\sigma_{A_0}\in \states(A_0)} \min\{\lambda>0,\ {\rho}\le\lambda\sigma_{A_0}\otimes I_{A_1}\}=
2^{-H_{\min}(A_1|A_0)_{{\rho}}}
\end{equation}
where  $H_{\min}$ denotes the conditional min-entropy \cite{renner2008security,konig2009theoperational}. We also have
\begin{equation}\label{eq:cond_min_op}
\|{\rho}\|^\diamond_{A_1|A_0}= \max_{\alpha\in \Ce(A_0,A_1)}\<{\rho},\alpha\>= \max_{\alpha\in \Ce(A_0,A_1)}\llangle I_{A_1}|
(\alpha\otimes id)({\rho})|I_{A_1}\rrangle.
\end{equation}
Note that the last 
 equality  corresponds to the operational interpretation of the conditional min-entropy as (up to multiplication by
 $d_A$)
the maximum fidelity with the maximally entangled state $\psi^{A_1}$ that can be obtained by applying a quantum channel to part $A_0$
of the state ${\rho}_{A_0A_1}$ \cite{konig2009theoperational}.

The following result follows easily from the first equality in \eqref{eq:cond_min}.
\begin{lemma}\label{lemma:dualnorm_inf}
Let $\rho\in \Be_+(A_0A_1)$. Then there is some $V\in\Be(A_0)$, $\Tr[VV^*]=1$, and $G\in\Be_+(A_0A_1)$ such that
\[
\rho=(\chi_V\otimes id)(G),\qquad \|G\|=\|\rho\|^\diamond_{A_1|A_0}.
\]
Here $\chi_V:=V\cdot V^*\in \Le_+(A_0,A_0)$  and $\|\cdot\|$ denotes the operator norm.

\end{lemma}

\subsubsection{Diamond 2-norm and conditional 2-min entropy}\label{sec:2diamond}

As we have seen, the set of superchannels $\Ce_2(A,A')$ is a subset of $\Ce(A'_0A_1,A'_1A_0)$. In fact, it is itself
a base section. More precisely, put 
 $\Le_2(A,A'):=\Le(A'_0A_1,A'_1A_0)$ with 
$\Le^+_2(A,A'):=\Le_+(A'_0A_1,A'_1A_0)$  and $B(\Le_2(A,A'))=\Ce_2(A,A')$, then  $\Le_2(A,A')\in \mathsf{BS}$.
Using the same identification of the dual space as before, we have $\Le^*(A,A')=B_h(A'_1A_0A'_0A_1)$ and the dual section is 
\[
B(\Le_2(A,A'))=\{\sigma*C_\gamma\otimes I_{A'_1},\ \sigma\in \states(A'_0R),\ \gamma\in \Ce(A_0R,A_1),\ R \mbox{ an
ancilla}\}
\]
Using the identification of the dual space as in  Remark \ref{rem:dual}, this corresponds to a set of superchannels  where the preprocessing  is a
replacement channel, of the form
\begin{center}
\includegraphics[align=c]{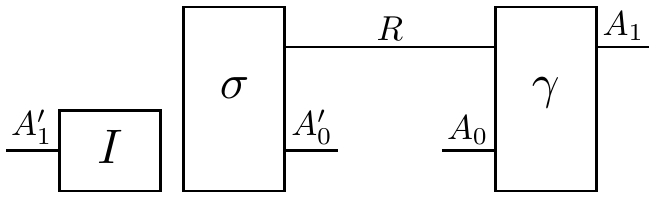}
\end{center}
The  norm $\|\cdot\|_{\Le_2(A,A')}$ is the distinguishability norm $\|\cdot\|_{2\diamond}$ for quantum networks, see
\cite{chiribella2008memory,gutoski2012onameasure} for the definition.
Let us denote 
\[
\|\cdot\|_{A'|A}^{2\diamond}:=\|\cdot\|_{ \Le^*_2(A,A')},
\]
  then  for
$\rho\in \Be_+(A'_0A_1A'_1A_0)$, we have
 \begin{align*}
\|{\rho}\|_{A'|A}^{2\diamond}&=\min_{\sigma,\gamma} \min\{\lambda>0,\ {\rho}\le\lambda(\sigma*C_{\gamma})_{A'_0A_1A_0}\otimes
I_{A'_1}\}=:2^{-H^{(2)}_{\min}(A'|A)_{{\rho}}}\\
&= \max_{\Theta\in \Ce_2(A,A')}\<{\rho},\Theta\>= \max_{\Theta\in \Ce_2(A,A') }\llangle I_{A_0A'_1}|
(\Theta\otimes id)(\rho)|I_{A_0A'_1}\rrangle.
\end{align*}
Here $H_{\min}^{(2)}$ will be called the conditional 2-min entropy. Note that this quantity coincides with the extended
conditional min-entropy of \cite{gour2019comparison}
but we prefer the present notation since  it can be extended to any $N\in \mathbb N$ in an obvious way using the set
of $N$-combs, see also  \cite{chiribella2016optimal}.
The last equality shows an operational interpretation as the maximum fidelity (again up to multiplication by the
dimension) with the maximally entangled state $\psi^{A_0B_1}$
that can be obtained by applying a structured quantum channel to the part $A'_0A_1$ of ${\rho}_{A'_0A_1A'_1A_0}$ as depicted in
the diagram 
\begin{center}
\includegraphics[align=c]{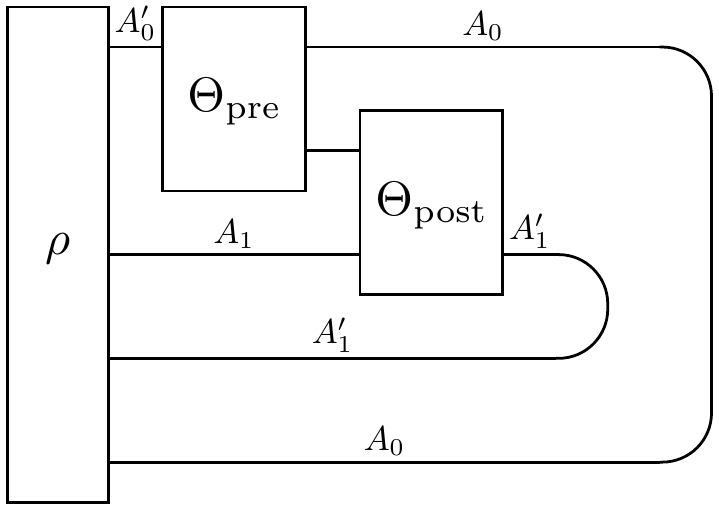}
\end{center}

\subsubsection{Guessing games}\label{sec:guessing}

Let  $\Ee=\{\lambda_i,\rho_i\}_{i=1}^k$ be an ensemble of states 
 $\rho_i\in \states(A)$ and
prior probabilities $\lambda_i$. Quantum measurements with $k$ outcomes are given by 
operators  $M=\{M_1,\dots,M_k\}$, where $M_j\in \Be_+(A)$, $\sum_j M_j=I_A$, the set of all such measurements for the
system $A$  will be denoted by $\Me_k(A)$.  It is well known that the optimal success probability $P_{\succ}(\Ee)$ is
related to the conditional min entropy as follows, \cite{konig2009theoperational}. 
Let us define the
quantum-classical state $\rho_\Ee=\sum_i\lambda_i\rho_i\otimes |i\>\<i|\in \states(AR)$, where $d_R=k$. 
Then for any channel $\alpha\in \Ce(A,R)$, we have 
\begin{equation}\label{eq:Psucc_e_m}
\<\rho_\Ee,\alpha\>=P_{\succ}(\Ee,M)=\<\rho_\Ee, \Phi_M\>,
\end{equation}
where $M\in \Me_k(A)$, $M_i=\alpha^*(|i\>\<i|)$ and  $\Phi_M\in \Ce(A,R)$ is the  quantum-to-classical (q-c) channel given by 
$\Phi_M(\sigma)=\sum_i \Tr[\sigma M_i]|i\>\<i|$.  It follows that
\begin{equation}\label{eq:ensemble}
\|\rho_\Ee\|_{R|A}^\diamond=P_{\succ}(\Ee),
\end{equation}
see also Example \ref{ex:measurements}. 
It was proved in \cite[Prop. 2]{jencova2016isit} that the dual norm $\|\cdot\|_{R|A}^\diamond$ can be interpreted as a success probability not only for
quantum-classical states. Since this result and the related constructions will be repeatedly used below, 
we give the proof here.

\begin{lemma}\label{lemma:dualnorm_Psuc} For any state $\rho\in \states(AR)$ there is
 an ensemble $\Ee^R_\rho$ on $AR$ such that 
%\begin{equation}\label{eq:dualnorm_Psuc}
\[
\|\rho\|^\diamond_{R|A}={d_R}P_{\succ}(\Ee^R_\rho).
\]
%\end{equation}
\end{lemma}

The proof is based on a relation between quantum channels $A\to R$ and measurements on $AR$, close to the Choi
representation of channels.  Let $\{U^R_1,\dots,U^R_{d_R^2}\}$ be the group of generalized Pauli unitaries on $R$ and
let $\Ue^R_x$
denote the conjugation by $U^R_x$, so that we have 
\[
\sum_x \Ue_x^R(X)=d_R\Tr[X]I_R.
\]
Let 
\[
\mathsf{B}^R_x:= (\Ue_x^R\otimes id)(\psi^R)=d_R^{-1}|U_x^{R}\doublek U_x^R|
\]
then
$\mathsf{B}^R=\{\mathsf{B}^R_1,\dots,\mathsf{B}^R_{d_R^2}\}$ defines the Bell measurement on $RR$.
For a channel $\beta\in \Ce(A,R)$, let $M^\beta\in \Me_{d_R^2}(AR)$ be defined as
\begin{equation}\label{eq:channel_meas1}
M^\beta_x=(\beta^*\otimes id)(\mathsf{B}^R_x),\quad x=1,\dots,d_R^2.
\end{equation}
Conversely, for any $M\in \Me_{d_R^2}(AR)$, let $\beta^M\in \Ce(A,R)$ be the channel obtained from the measurement $M$
 in the teleportation scheme, that is
\begin{equation}\label{eq:meas_channel}
\beta^M(\sigma)=\sum_x \Ue^R_x(\ptr_{A\tilde R}[(\sigma \otimes \psi^R)(M_x \otimes I_R)]),\quad \sigma\in \states(A)
\end{equation}
(here we take $\tilde R\simeq R$ and view $\psi^R$ as a state on $\tilde RR$ and $M$ as a measurement on $A\tilde R$). 
It is easily seen that we have $M^{\beta^M}=M$ and $\beta^{M^\beta}=\beta$.
Note that we have $\beta^M=\sum_x \beta^M_x$, where $\{\beta^M_x\}$ is an instrument whose elements are depicted as
\begin{center}
\includegraphics{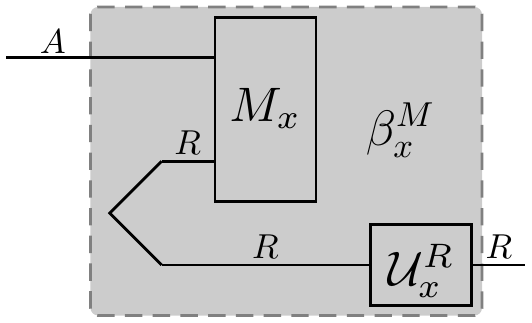}
\end{center}

\begin{proof}[Proof of Lemma \ref{lemma:dualnorm_Psuc}]

For $\rho\in \states(AR)$ we introduce the equiprobable ensemble 
\[
\Ee^R_\rho=\{d_R^{-2}, \rho_x\}_{x=1}^{d_R^2},\quad \rho_x=(id_A\otimes \tilde \Ue_x^R)(\rho),
\]
where $\tilde \Ue^R_x$ denotes the conjugation by $(U_x^R)^{\mathsf{T}}$.
Note that for  any channel $\beta\in \Ce(A,R)$ and $x=1,\dots,d_R^2$, we have
\begin{center}
\begin{minipage}[c]{0.1\textwidth}
\centering
$d_R^{-1}\<\rho,\beta\> =$
\end{minipage}
\begin{minipage}[c]{0.25\textwidth}
\centering
\includegraphics{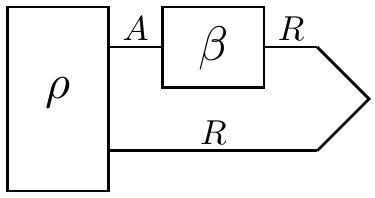}
\end{minipage}
\begin{minipage}[c]{0.01\textwidth}
 = 
\end{minipage}
\begin{minipage}[c]{0.4\textwidth}
\centering
\includegraphics{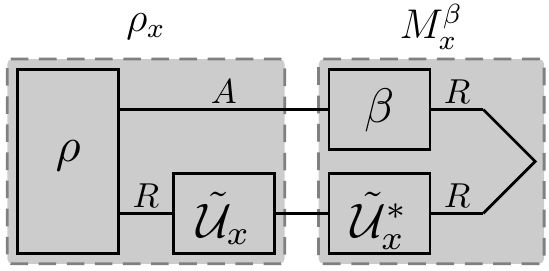}
\end{minipage}
\end{center}
 Multiplying by the probability $d_R^{-2}$ and  summing up over $x$ we obtain
\begin{equation}\label{eq:channel_meas}
d_R^{-1}\<\rho,\beta\>= P_{\succ}((\beta\otimes
id_R)(\Ee^R_\rho),\mathsf{B}^R)=  P_{\succ}(\Ee^R_\rho,M^\beta).
\end{equation}
Also conversely, it is readily checked that for any measurement $M\in \Me_{d_R^2}(AR)$ we have
\begin{equation}\label{eq:meas_channel1}
P_{\succ}(\Ee_\rho,M)=\sum_x \Tr[\rho_x M_x] = d_R^{-1}\sum_x \<\rho, \beta^M_x\>=d_R^{-1}\<\rho,\beta^M\>.
\end{equation}
Using \eqref{eq:cond_min_op}, this proves Lemma \ref{lemma:dualnorm_Psuc}.

\end{proof}

As an application, we have the following expression for the diamond norm distance of quantum channels in terms of the success probabilities  in guessing games. %%something in the intro...+ picture 

\begin{coro}\label{coro:psuc}
Let $\Phi_1,\Phi_2\in \Ce(A_0,A_1)$. Then 
\[
\sup_{\Ee, M}
\frac{P_{\succ}((\Phi_1\otimes id_R)(\Ee),M)-P_{\succ}((\Phi_2\otimes id_R)(\Ee),M)}{P_{\succ}(\Ee)}=\frac12\|\Phi_1-\Phi_2\|_\diamond,
\]
where the supremum is taken over all ensembles $\Ee$ on $A_0R$, all measurements $M$ on $A_1R$ and any ancilla $R$.
Moreover, the supremum is attained with $R\simeq A_1$ and the Bell measurement $M=\mathsf{B}^{A_1}$.
\end{coro}

\begin{proof} 
Let $s$ denote the supremum on the left hand side of the equality to be proved. 
Let $\Ee=\{\lambda_j,\rho_j\}_{j=1}^k$ be an ensemble on $A_0R$ and $M\in \Me_k(A_1R)$. Put 
$\rho=\rho_\Ee\in \states(A_0RR')$, with $d_{R'}=k$. Then by \eqref{eq:Psucc_e_m}
\[
P_{\succ}((\Phi_i\otimes id)(\Ee),M)=P_{\succ}(\Ee, (\Phi_i^*\otimes id)(M))=\<\rho, \Phi_{(\Phi^*_i\otimes id)(M)}\>.
\]
$i=1,2$. Note that the q-c channel $\Phi_{(\Phi_i^*\otimes id)(M)}=\Phi_M\circ (\Phi_i\otimes id_R)$ and
therefore
\begin{align*}
P_{\succ}((\Phi_1\otimes id)(\Ee), M)&-P_{\succ}((\Phi_2\otimes id)(\Ee),M)=
\<\rho_\Ee,\Phi_M\circ(\Phi_1\otimes id_R)-\Phi_M\circ(\Phi_2\otimes id_R)\>
\\
&\le \frac12 \|\rho_\Ee\|_{R|A}^\diamond \|\Phi_M\circ[(\Phi_1-\Phi_2)\otimes id_R]\|_\diamond\le\frac12
P_{\succ}(\Ee)\|\Phi_1-\Phi_2\|_\diamond,
\end{align*}
here we used Proposition \ref{prop:basesec}. This shows that $s \le \frac12\|\Phi_1-\Phi_2\|_\diamond$.

For the converse, note that by Proposition \ref{prop:basesec} (iii) and (iv), we have
\begin{equation}\label{eq:diamond_psuc}
\frac12\|\Phi_1-\Phi_2\|_\diamond= \sup_{\rho\in \Be_+(A_0A_1)}\frac{\<\rho,
\Phi_1-\Phi_2\>}{\|\rho\|^\diamond_{A_1|A_0}}= 
\sup_{\rho\in \states(A_0A_1)}\frac{\<\rho,
\Phi_1-\Phi_2\>}{\|\rho\|^\diamond_{A_1|A_0}}. 
\end{equation}
Let now $R\simeq A_1$ and $M=\mathsf{B}^{A_1}$. For $\rho\in \states(A_0A_1)$, we take the ensemble $\Ee_\rho^{A_1}$ on 
 $A_0A_1$. By \eqref{eq:diamond_psuc}, \eqref{eq:channel_meas} and Lemma \ref{lemma:dualnorm_Psuc}, we obtain 
\[
\frac12\|\Phi_1-\Phi_2\|_\diamond=\sup_{\rho\in \states(A_0A_1)}\frac{P_{\succ}((\Phi_1\otimes
id)(\Ee_\rho^{A_1}),\mathsf{B}^{A_1})-P_{\succ}((\Phi_2\otimes id)(\Ee_\rho^{A_1}),\mathsf{B}^{A_1})
}{P_{\succ}(\Ee_\rho^R)}\le s.
\] 
 This finishes the proof.

\end{proof}

\begin{rem} The above construction implies another operational interpretation of $\|\rho\|^\diamond_{R|A}$. 
To see this, let $\Ee=\{\lambda_i, \Phi_i\}$ be an ensemble of quantum channels, $\Phi_i\in \Ce(R,R)$. The guessing
procedures for ensembles of channels can be described by pairs
$(\rho,M)$, consisting of an input state  $\rho\in \states(AR)$ with some ancilla $A$ and $M$ is a measurement on $AR$,
such triples are also called
quantum testers \cite{chiribella2008memory}. The average success probability for the tester $(\rho,M)$
is then
\[
P_{\succ}(\Ee,\rho,M):=P_{\succ}(\Ee(\rho),M)
\]
where $\Ee(\rho)=\{\lambda_i, (\Phi_i\otimes id)(\rho)\}$. 

Any state  $\rho\in \states(AR)$ can be seen as the input state of some tester. We claim that 
the norm $\|\rho\|^\diamond_{R|A}$ can be interpreted as ($d_R$ times) the maximal success
probability that can be obtained by all testers with input state $\rho$ for equiprobable ensembles 
$\Ee=\{d_R^{-2},\Phi_i\}_{x=1}^{d_R^2}$ of unital channels $\Phi_x\in \Ce(R,R)$.
Indeed, let  $M\in \Me_{d_R^2}(AR)$ be any measurement, we have
\[
P_{\succ}(\Ee,\rho,M)=d_R^{-2}\sum_x  \Tr[(id\otimes \Phi_x)(\rho)M_x]=  d_R^{-1}\Tr[\rho C],
\]
where $C=d_R^{-1}\sum_i (id\otimes \Phi_x^*)(M_x)\in \Be_+(AR)$. Since $\Phi_x^*$ is trace preserving, we see that 
$\ptr_R[C]=d_R^{-1}\ptr_R[I]=I_A$. Hence there is a channel $\alpha\in \Ce(R,A)$ such that $C=C_{\alpha^*}$. Finishing the
above computation, we obtain
\[
P_{\succ}(\Ee,\rho,M)=d_R^{-1}\Tr[\rho C_{\alpha^*}]=d_R^{-1}\<\rho,\alpha\>\le d_R^{-1}\|\rho\|_{R|A}.
\]
As we have seen, equality is attained for $\Phi_x=\tilde{\mathcal U}_x^R$.

\end{rem}

\subsection{General comparison theorems for quantum channels}

We are now ready to apply  the results of Section \ref{sec:comp_gpt} to the comparison of 
 a pair of quantum channels $\Phi_1$ and $\Phi_2$. For this, we consider a subcategory  $\Fe$ in $\mathsf{BS}$
whose objects  are some  spaces of channels (as in Sec.
\ref{sec:diamond}) and morphisms  between them are given by some convex subsets of superchannels, (so $\Fe$ is in fact a
convex subcategory of the category of quantum
channels with superchannels).

If $\Le(A_0,A_1)$ is an object in $\Fe$ we will say that $A=A_0A_1$ is an  input-output space admissible in $\Fe$. 
We will use the notation 
\[
\Fe(A,A'):=\Fe(\Le(A_0,A_1), \Le(A'_0,A'_1))\subseteq \Ce_2(A,A')
\]
for a pair of admissible spaces $A$, $A'$.

Let $A$, $A'$ be admissible in $\Fe$  and let $\Phi_1\in \Ce(A_0,A_1)$, $\Phi_2\in\Ce(A'_0,A'_1)$.
The $\Fe$-conversion distance  becomes 
\[
\delta_\Fe(\Phi_1\|\Phi_2)=\inf_{\Lambda\in \Fe(A,A')}\|\Lambda(\Phi_1)-\Phi_2\|_\diamond.
\]
{
For $\rho\in \Be_+(A'_0A_1A'_1A_0)$, we define
\[
\|\rho\|^\Fe_{A'|A}:=\sup_{\Theta\in \Fe(A,A')} \<\rho,\Theta\>.
\]
This notation may be somewhat misleading, since for a general  subcategory $\Fe$
there is no guarantee that $\|\cdot\|^\Fe_{A'|A}$ can be extended to a norm. 
However, there are some choices that lead to the norms introduced in Sections \ref{sec:diamond} and \ref{sec:2diamond}. Indeed, with the choice of a subcategory where all objects are spaces of states (that is, all admissible input spaces are
$A_0=1$) and the morphisms are all (super)channels between them, we obtain
$\|\cdot\|^\Fe_{A'|A}= \|\cdot\|^\diamond_{A'_1|A_1}$.  If $\Fe$ coincides with the category of quantum channels with superchannels, we similarly get 
 $\|\rho\|^\Fe_{A'|A}=\|\rho\|^{2\diamond}_{A'|A}$. Since 
$\Fe(A,A')\subseteq \Ce_2(A,A')\subseteq \Ce(A'_0A_1,A'_1A_0)$, the following inequalities are immediate:
\[
\|\rho\|^\Fe_{A'|A}\le \|\rho\|^{2\diamond}_{A'|A}\le \|\rho\|^\diamond_{A'_1A_0|A'_0A_1}.
\]
Moreover, the quantity
\[
H_{\min}^\Fe(A'|A)_\rho:=-\log(\|\rho\|^\Fe_{A'|A})
\]
can be seen as a modified conditional min entropy, since it coincides with $H_{\min}$ and $H^{(2)}_{\min}$ in the above
cases.

More generally, it may happen that
 $\Fe(A,A')$ is a base section in $\Le_2(A,A')$, possibly with respect to some subcone 
$\Le^+_\Fe(A,A')\subseteq \Le_2^+(A,A')$. 
Then $\|\cdot\|^\Fe_{A'|A}$ coincides with the base section norm with respect to the dual section, see Proposition
\ref{prop:basesec} (iii). Consequently, $\|\cdot\|_{A'|A}^\Fe$ satisfies the properties in Proposition \ref{prop:basesec},
so that there is a  dual expression for this norm. {In particular, it can be expressed  as a conic
program. } 

\begin{ex}\label{ex:ppt_sep_ns} Let $\Fe$ be the subcategory whose objects are spaces $\Le(A_0B_0,A_1B_1)$ 
of bipartite channels and the 
morphisms in $\Fe$ are PPT-superchannels defined by the property that they stay completely positive when pre and postcomposed by the partial transpose
supermap, see \cite{gour2019theentanglement} for more details. In this case, $\Fe(AB,A'B')$ is a base section in $\Le_2$
with respect to the subcone $\Le_{2,PPT}^+\subsetneq \Le_2^+$ of 
PPT-supermaps, \cite[Definition V.2]{gour2019theentanglement}. Similarly,  if the morphisms are given by 
separable superchannels, characterized by the condition that the Choi matrix is separable, then $\Fe(AB,A'B')$ is a base
section with respect to the cone
 $\Le_{2,SEP}^+$ of separable supermaps. Note that
these sets of superchannels might be different from their restricted variants that will be considered below.

Another such example is the subcategory with spaces of channels as objects and morphisms given by 
  no-signaling superchannels which can be easily described by properties of their Choi matrices: besides the
condition \eqref{eq:comb} of a 2-comb they also satisfy an analogical condition with input and output spaces exchanged. 
In this case,  $\Fe(A,A')$ is a base section in  the cone $\Le^+_2(A,A')$.

\end{ex}

For any choice of $\Fe$,  $\|\cdot\|^\Fe_{A'|A}$ is obviously convex and monotone under morphisms in $\Fe$, more precisely,
 for any $\Lambda\in \Fe(A',A'')$, we have
\[
\|\rho\|^\Fe_{A'|A}\ge \|\rho*C_\Lambda\|^\Fe_{A''|A}.
\] 
Further properties depend on the details of the structure of $\Fe$, e.g. with respect to the tensor products.
 
The importance of these quantities is seen from the following general comparison  theorem for quantum channels, which is 
 a straightforward reformulation of Theorem
\ref{thm:rand} from the GPT setting.

\begin{thm}\label{thm:rand_chans}  Let $\epsilon\ge0$. The following are equivalent.
\begin{enumerate}
\item[(i)] $\delta_\Fe(\Phi_1\|\Phi_2)\le \epsilon;$
\item[(ii)] for any $\rho\in \states(A'_0A'_1)$, there is some $\Theta\in \bar{\Fe}(A,A')$ such that
\[
\<\rho,\Phi_2\>\le \<\rho, \Theta(\Phi_1)\>+ \frac{\epsilon}2\|\rho\|^\diamond_{A'_1|A'_0}
\]
(here $\bar{\Fe}(A,A')$ is the closure of the set $\Fe(A,A')$);
\item[(iii)] for any $R=R_0R_1$ admissible in $\Fe$  and any $\rho\in \states(R_0R_1)$,
\begin{align*}
\|\rho*C_{\Phi_2}\|^\Fe_{R|A'}\le \|\rho*C_{\Phi_1}\|^\Fe_{R|A}+ \frac{\epsilon}2\|\rho\|^\diamond_{R_1|R_0}.
\end{align*}

\end{enumerate}
Moreover, in (iii) it is enough to assume $R_0\simeq A'_0$, $R_1\simeq A'_1$.

\end{thm}

\begin{proof} We only need to observe that for any $\Theta\in \Fe(A,R)$ and $\rho\in \states(R)$, we have
\[
\<\rho,\Theta(\Phi_1)\>=\rho*C_{\Theta(\Phi_1)}=(\rho*C_{\Phi_1})*C_\Theta=\<\rho*C_{\Phi_1},\Theta\>,
\]
similarly for $\Phi_2$.
The proof now follows directly from Theorem \ref{thm:rand} and the definition of $\|\cdot\|^\Fe_{R|A}$
($\|\cdot\|^\Fe_{R|A'}$).

\end{proof}

\begin{rem}\label{rem:F_link} There is an ambiguity around  $\rho*C_{\Phi_i}$ that might cause some confusion: it may
happen that parts
of $R_0$ or $R_1$ coincide with some parts of the input or output spaces of $\Phi_i$, so it is unclear how to apply the
link product. In some cases, such as in some of the sections below, the subcategory $\Fe$ does not permit any processing on some
parts of the input and/or output spaces. In this case, these parts are always fixed and are viewed as
the same, so that they are connected in the link product. In diagram, if the fixed input of the channels is $B_0$ and
the fixed output is $B_1$, we obtain in this case
\begin{center}
\begin{minipage}[c]{0.15\textwidth}
\centering
$\<\rho,\Theta(\Phi)\> =$
\end{minipage}
\begin{minipage}[c]{0.5\textwidth}
\centering
\includegraphics{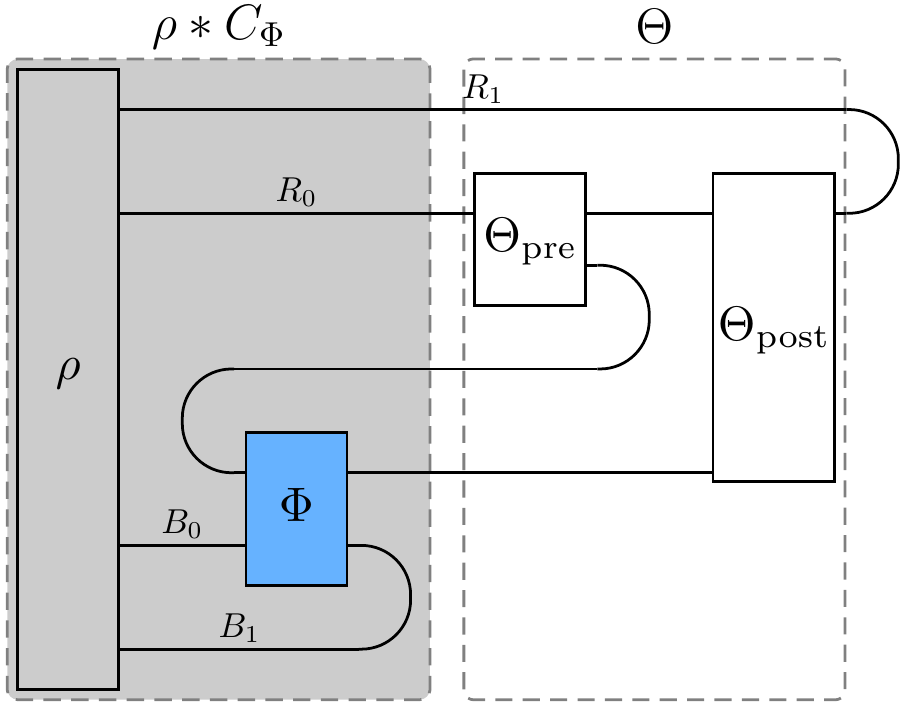}
\end{minipage}
\begin{minipage}[c]{0.2\textwidth}
\centering
$=\<\rho*C_\Phi,\Theta\>$.
\end{minipage}
\end{center}
In all other cases all the involved spaces are treated as independent, so
that the link product is in fact the tensor product, $\rho*C_{\Phi_i}=\rho\otimes C_{\Phi_i}$, in diagram
\begin{center}
\begin{minipage}[c]{0.15\textwidth}
\centering
$\<\rho,\Theta(\Phi)\> =$
\end{minipage}
\begin{minipage}[c]{0.4\textwidth}
\centering
\includegraphics{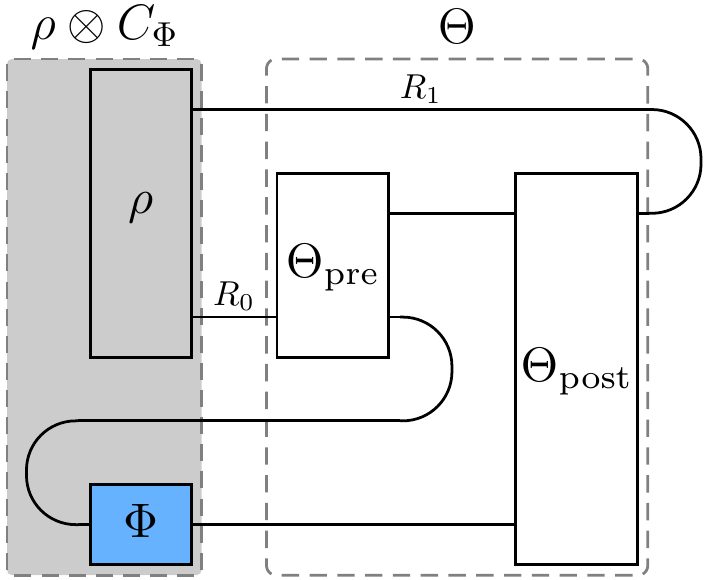}
\end{minipage}
\begin{minipage}[c]{0.15\textwidth}
\centering
$=\<\rho\otimes C_\Phi,\Theta\>$.
\end{minipage}
\end{center}

\end{rem}

The questions of simulability/convertibility of quantum channels naturally appear in the quantum resource theory of
processes, \cite{liu2020operational,liu2019resource,gour2019theentanglement,gour2020dynamical}. In this case, the subset $\mathcal
O(A_0,A_1):=\Fe(1,A)\subseteq \Ce(A_0,A_1)$ is the set of free channels and $\Fe(A,A')$ is the set of free
superchannels, that is, transformations that can be performed at no cost. 
The subcategory is usually assumed to have further properties, such as that it behaves well under tensor
products (i.e. it is a symmetric monoidal subcategory  in the category of quantum channels with
superchannels). Note that a convex symmetric monoidal category is called a convex resource theory in
\cite{coecke2016amathematical}.

   Assume that the resource theory is such that all free channels can be converted one into another
by free superchannels, then for each $\rho\in \states(R_0R_1)$ the map 
\[
\varphi_\rho: \Phi\mapsto \|\rho*C_\Phi\|^\Fe_{R|A}
\]
is constant over $\Phi\in \mathcal O$, namely $\varphi_\rho(\Phi)=\sup_{\Psi\in \mathcal O}\<\rho,\Psi\>$ for any $\Phi\in
\mathcal O$. Therefore each $\varphi_\rho$ can be easily normalized to be 0 on $\mathcal O$. Furthermore, if any channel
$\Phi$ can be converted into a channel arbitrarily close to $\mathcal O$ by some elements in ${\Fe}$, then Theorem \ref{thm:rand_chans} 
implies that $\varphi_\rho(\Phi)\ge \sup_{\Psi\in \mathcal O}\<\rho,\Psi\>$ and the set of normalized $\varphi_\rho$ becomes a complete family of resource monotones. In fact,
the case $\epsilon =0$ is closely related to \cite[Theorem III.3]{gour2019theentanglement}.

Specific examples of the subcategory $\Fe$ will be studied in detail in the next sections: postprocessings and 
preprocessings of quantum channels, processings of bipartite channels by LOCC and by partial superchannels.
In all these cases, the superchannels in $\Fe$ consist of 
pre- and postprocessings belonging to some 
specified families of channels $\mathcal C_{\pre}$ and $\mathcal C_{\post}$. More precisely, for $R=R_0R_1$, $S=S_0S_1$
admissible in $\Fe$, any superchannel $\Theta\in \Fe(R,S)$ has the form $\Theta=\Lambda_{\pre}*\Lambda_{\post}$,
where $\Lambda_{\pre}\in \Ce_{\pre}(S_0, UR_0)$ and $\Lambda_{\post}\in \Ce_{\post}(UR_1,S_1)$, where $U$ is some ancilla.
 To ensure that $\Fe$ is a convex subcategory, we have to assume that
for any input-output spaces $R,S,T$ admissible in $\Fe$ and any ancillas $U,V$ available in $\Ce_{\pre}$ and $\Ce_{\post}$, we have
\begin{enumerate}
\item[(a)] $id_{R_0}\in \Ce_{\pre}(R_0,R_0)$, $id_{R_1}\in \Ce_{\post}(R_1,R_1)$;
\item[(b)] both $\Ce_{\pre}$ and $\Ce_{\post}$ are closed
under composition, that is, 
\[
\alpha\in \Ce_{\pre}(S_0,UR_0),\ \alpha'\in \Ce_{\pre}(T_0,VS_0) \ \implies\ (id_V\otimes \alpha)\circ \alpha'\in 
\Ce_{\pre}(T_0,VUR_0)
\]
and 
\[
\beta\in \Ce_{\post}(UR_1,S_1),\ \beta'\in \Ce_{\post}(VS_1,T_1)\ \implies \beta'\circ(id_V\otimes \beta)\in
\Ce_{\post}(VUR_1,T_1);
\]
\item[(c)] the sets $\Ce_{\pre}(S_0,UR_0)$ and $\Ce_{\post}(UR_1,S_1)$ are convex.
\end{enumerate}
Further conditions may be needed to ensure convexity, such as
\begin{enumerate}
\item[(d)] $\Ce_{\pre/\post}$  are closed under appending/discarding  a classical register, more precisely, that  we have
\[
|i\>\<i|\otimes (\cdot) \in \mathcal C_{\pre}(R_0,UR_0),\qquad \sum_i \<i|\cdot  |i\>\otimes \Phi_i \in
\Ce_{\post}(UVR_1,S_1),\ \text{if } \Phi_i\in \mathcal C_{\post}(VR_1,S_1),\ \forall i.
\]

\end{enumerate}

 We will write 
\[
\Fe=\mathcal C_{\pre}*\mathcal C_{\post}
\]
to emphasize that the morphisms in $\Fe$ have this form.
This  is a natural assumption in the framework of resource theories, where $\mathcal C_{\pre}=\mathcal C_{\post}=\mathcal
O$ is the set of free channels \cite{coecke2016amathematical,gour2019theentanglement}. We will study the subcategory of LOCC
superchannels, where the objects are spaces of bipartite channels and $\mathcal O=\Ce_{LOCC}$. We may similarly consider
 the category of restricted PPT superchannels, where $\mathcal O=\Ce_{PPT}$, or restricted separable superchannels 
where $\mathcal O=\Ce_{SEP}$, \cite{gour2019theentanglement, gour2020dynamical}. 

We will use the decomposition $\Fe=\Ce_{\pre}*\Ce_{\post}$ to characterize $\delta_\Fe$ 
 by success probabilities in  modified guessing games of the form  depicted in the diagram
\begin{center}
\includegraphics{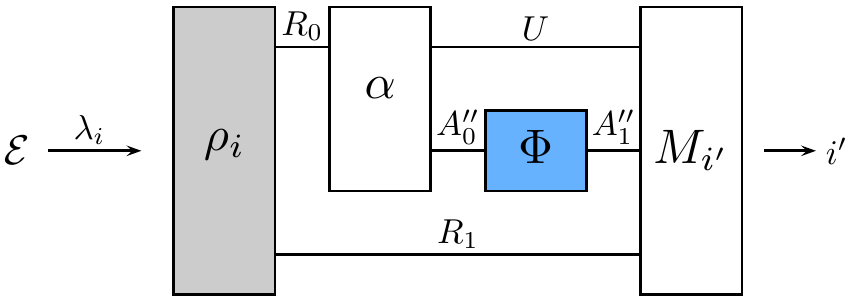}
\end{center}

Here $\Ee=\{\lambda_i,\rho_i\}$ is an ensemble on an admissible space $R=R_0R_1$, a preprocessing $\alpha$ can be chosen from 
$\Ce_{\pre}(R_0,UA_1'')$,   $M$ is a measurement chosen from a family $\Me_{\post}(UA_1''R_1)$ of allowed measurements
 and $A_0''$, $A_1''$ are the input and output spaces of  $\Phi$, which  is either $\Phi_1$ or $\Phi_2$.  The optimal success probability with this scheme is 
\[
P_{\succ}^{\Ce_{\pre}(R_0,\cdot A_0''),\Phi,\Me_{\post}}(\Ee):= \sup_{\alpha\in \Ce_{\pre}, M\in \Me_{\post}} 
P_{\succ}((\alpha\otimes id)(\Ee),(\Phi^*\otimes id)(M)),
\] 
here $\Ce_{\pre}(R_0,\cdot A_0'')$ indicates that we allow any ancilla $U$ that is permitted by the structure of
$\Ce_{\pre}$ and $\Me_{\post}$.
We will also  assume  that the family $\Me_{\post}$ is closed under  preprocessings by $\Ce_{\post}$, more precisely, for any $R,S,T$ admissible in $\Fe$ and any
ancillas $U,V$, we have
\[
\beta\in \Ce_{\post}(UR_1,S_1),\ M\in \Me_{\post}(VS_1T_1)\ \implies \ (id_{VT_1}\otimes \beta^*)(M)\in
\Me_{\post}(UVR_1T_1).
\]
%diagram

\begin{thm}\label{thm:rand_chans_psuc} Let $\Fe=\Ce_{\pre}*\Ce_{\post}$ and let $\Me_{\post}$ be a family of
measurements closed under preprocessings by $\Ce_{\post}$. 
If $\delta_\Fe(\Phi_1\|\Phi_2)\le \epsilon$, then for any $R=R_0R_1$ admissible in $\Fe$ and any ensemble $\Ee$ on $R_0R_1$ we
have
\[
P_{\succ}^{\Ce_{\pre}(R_0,\cdot A'_0),\Phi_2,\Me_{\post}}(\Ee)\le
P_{\succ}^{\Ce_{\pre}(R_0,\cdot A_0),\Phi_1,\Me_{\post}}(\Ee)+\frac{\epsilon}2P_{\succ}(\Ee).
\]
Moreover, if the Bell measurement $\mathsf{B}^{A'_1}\in \Me_{\post}(A'_1A'_1)$ and for any measurement 
$M\in \Me_{\post}(\cdot A_1A'_1)$ with
$d_{A'_1}^2$ outcomes we have  $\beta^M\in \Ce_{\post}$, then the converse also holds, with 
 $R_0\simeq A'_0$ and $R_1\simeq A'_1$.
\end{thm}

\begin{proof} Assume that there is some $\Lambda\in \bar{\Fe}(A,A')$ such that $\|\Lambda(\Phi_1)-\Phi_2\|_\diamond\le
\epsilon$. Let $\Ee$ be an ensemble on $R_0R_1$ and let $\alpha\in \Ce_{\pre}(R_0,UA'_0)$, $M\in \Me_{\post}(UA'_1R_1)$.
Then
\[
\|\Lambda(\Phi_1)\circ\alpha-\Phi_2\circ\alpha\|_\diamond=\|(\Lambda(\Phi_1)-\Phi_2)\circ\alpha\|_\diamond \le \epsilon
\]
and using  Corollary \ref{coro:psuc}, we have
\[
P_{\succ}((\alpha\otimes id)(\Ee),(\Phi^*_2\otimes id)(M))\le P_{\succ}((\alpha\otimes id)(\Ee),(\Lambda(\Phi_1)^*\otimes
id)(M))+\frac{\epsilon}2P_{\succ}(\Ee).
\]
Let $\Lambda=\Lambda_{\pre}*\Lambda_{\post}$, with $\Lambda_{\pre}\in \Ce_{\pre}(A'_0,VA_0)$, $\Lambda_{\post}\in
\Ce_{\post}(VA_1,A'_1)$. Then $\alpha'=\Lambda_{\pre}\circ\alpha\in \Ce_{\pre}(R_0, UVA_0)$ and $M'= \Lambda_{\post}^*(M)\in
\Me_{\post}(UVA_1R_1)$, so that 
\[
P_{\succ}((\alpha\otimes id)(\Ee),(\Lambda(\Phi_1)^*\otimes id)(M))=P_{\succ}((\alpha'\otimes id)(\Ee),(\Phi_1^*\otimes
id)(M'))\le P_{\succ}^{\Ce_{\pre}(R_0,\cdot A_0),\Phi_1,\Me_{\post}}.
\] 
%tensoring with identities, diagrams.

We now prove the converse, assuming the two additional conditions. Suppose that the inequality holds with $R_0\simeq
A'_0$ and $R_1\simeq A'_1$. Let $\rho\in \states(A'_0A'_1)$ and let $\Ee=\Ee^{A'_1}_\rho$ be as in Section \ref{sec:guessing}, 
then by \eqref{eq:channel_meas}, we have
\begin{align*}
\<\rho,\Phi_2\>&=d_{A'_1}P_{\succ}(\Ee,(\Phi_2^*\otimes id)(\mathsf{B}^{A'_1})) \le
d_{A'_1}P_{\succ}^{\Ce_{\pre}(A'_0,\cdot A'_0),\Phi_2,\Me_{\post}}(\Ee)\\
&\le d_{A'_1}P_{\succ}^{\Ce_{\pre}(A'_0,\cdot A_0),\Phi_1,\Me_{\post}}(\Ee)+\frac{\epsilon}2d_{A'_1}P_{\succ}(\Ee),
\end{align*}
here we have used that $id_{A'_0}\in \Ce_{\pre}(A'_0,A'_0)$ and $\mathsf{B}^{A'_1}\in \Me_{\post}(A'_1A'_1)$. Choose any 
 $\alpha\in \Ce_{\pre}(A'_0,UA_0)$ and $M\in \Me_{\post}(UA_1A'_1)$ and consider the success probability 
$P_{\succ}((\alpha\otimes id)(\Ee),(\Phi_1^*\otimes id)(M))$.  Since the channel $\Phi_1\circ \alpha$ acts only on 
$A'_0$, we wee that 
\[
P_{\succ}((\alpha\otimes id)(\Ee), (\Phi_1^*\otimes id)(M))=P_{\succ}(\Ee^{A'_1}_{\tilde \rho}, M),
\]
with $\tilde \rho=(\Phi_1\circ\alpha\otimes id)(\rho)$. 
By \eqref{eq:meas_channel1} and the assumption, we have $\beta =\beta^M\in \Ce_{\post}(UA_1,A'_1)$ and 
\[
d_{A'_1}P_{\succ}(\Ee^{A'_1}_{\tilde \rho}, M)=\<\tilde \rho,\beta\>=\<(\Phi_1\circ\alpha\otimes id)(\rho),\beta\>=\<\rho,
\Theta(\Phi_1)\>,
\]
with $\Theta=\alpha*\beta\in \Fe(A,A')$.    Using the definition of $\|\cdot\|^\Fe_{A'|A}$ and Lemma
\ref{lemma:dualnorm_Psuc} we now obtain
\[
\<\rho,\Phi_2\>\le \|\rho* C_{\Phi_1}\|^\Fe_{A'|A}+\frac{\epsilon}2\|\rho\|^\diamond_{A'_1|A'_0}.
\]
By the condition (ii) in Theorem \ref{thm:rand_chans} this finishes the proof.

\end{proof}

}

\subsection{Comparison by postprocessings}\label{sec:post}

 Comparison of quantum channels by postprocessings was already considered in
\cite{jencova2016isit} (see also \cite{jencova2015comparison} for a longer version with complete proofs), where
 quantum versions of the randomization theorem (Example \ref{ex:lecam}) were studied and   the results below were
obtained. We consider this case for completeness, just to show that they fit
into the setting of Theorem \ref{thm:rand_chans}.

Assume that the channels $\Phi_1\in \Ce(A_0,A_1)$ and $\Phi_2\in \Ce(A_0,A'_1)$ have the same input space $A_0$.
The postprocessing deficiency of $\Phi_1$ with respect to $\Phi_2$ was defined in \cite{jencova2016isit}  as 
\[
\delta_{\post}(\Phi_1\|\Phi_2):=\min_{\Lambda\in \Ce(A_1,A'_1)} \|\Lambda\circ \Phi_1-\Phi_2\|_\diamond.
\]
Let  $\Fe$ be  the subcategory with admissible spaces of the form $A_0R_1$, with  a fixed input system  $A_0$,  and morphisms $\Le(A_0,R_1)\to
\Le(A_0,S_1)$  given by postprocessings $\Phi\mapsto \Lambda\circ \Phi$ for some $\Lambda\in \Ce(R_1,S_1)$. The sought
approximation now becomes

 \begin{center}
\begin{minipage}[c]{0.3\textwidth}
\centering
\includegraphics{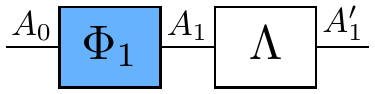}
\end{minipage}
\begin{minipage}[c]{0.01\textwidth}
 $\approx$
\end{minipage}
\begin{minipage}[c]{0.2\textwidth}
\centering
\includegraphics{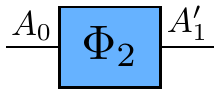}
\end{minipage}
 \end{center}

It is clear that $\delta_{\post}(\Phi_1\|\Phi_2)=\delta_\Fe(\Phi_1\|\Phi_2)$ and we may apply the results of the previous
section. Note also that $\Fe=\Ce_{\pre}*\Ce_{\post}$, where $\Ce_{\pre}=\{id_{A_0}\}$ and $\Ce_{\post}=\Ce$.  

\begin{thm}\label{thm:post} Let $\Phi_1\in \Ce(A_0,A_1)$, $\Phi_2\in \Ce(A_0,A'_1)$ and let $\epsilon\ge 0$. The
following are equivalent.
\begin{enumerate}
\item[(i)] $\delta_{\post}(\Phi_1\|\Phi_2)\le \epsilon$;
\item[(ii)] For any ancilla $R_1$ and  $\rho\in \states(A_0R_1)$, we have
\[
\|(\Phi_2\otimes id_{R_1})(\rho)\|_{R_1|A'_1}^\diamond\le \|(\Phi_1\otimes
id_{R_1})(\rho)\|_{R_1|A_1}^\diamond+\frac{\epsilon}2\|\rho\|_{R_1|A_0}^\diamond;
\]
\item[(iii)] For any ancilla $R_1$ and  any ensemble $\Ee$ on $A_0R_1$, we have
\[
P_{\succ}((\Phi_2\otimes id_{R_1})(\Ee))\le P_{\succ}((\Phi_1\otimes id_{R_1})(\Ee))+\frac{\epsilon}2P_{\succ}(\Ee).
\]
\end{enumerate}

Moreover, in (ii) and (iii), it is enough to use $R_1\simeq A'_1$. 

\end{thm}

\begin{proof}
Note that for any $\rho\in \states(A_0R_1)$ we have $\rho*C_{\Phi_i}=(\Phi_i\otimes id)(\rho)$ (see Remark
\ref{rem:F_link}). It follows
that
\[
\|\rho*C_{\Phi_2}\|^\Fe_{A_0R_1|A_0A'_1}=\sup_{\Lambda\in \Ce(A'_1,R_1)} \<(\Phi_2\otimes id)(\rho), \Lambda\>=
\|(\Phi_2\otimes id)(\rho)\|^\diamond_{R_1|A'_1}
\]
and similarly for $\Phi_1$. This shows the equivalence (i) $\iff$ (ii) and the fact that we may assume $R_1\simeq A'_1$,
by a direct application of Theorem \ref{thm:rand_chans}. The rest of the proof follows by Theorem
\ref{thm:rand_chans_psuc} with $\Me_{\post}$ being the
set of all measurements.

\end{proof}

As we have seen, the guessing games have a simple form here, with no preprocessings and no restrictions on the
measurement $M$:

\begin{center}
\includegraphics{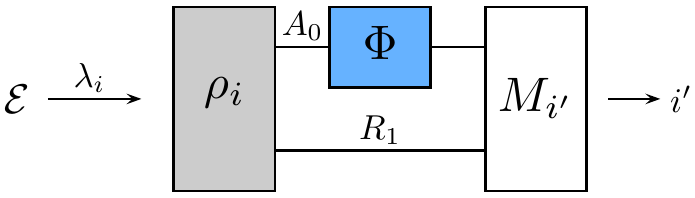}
\end{center}

 Note that classical-to-quantum channels can be identified with statistical experiments, see Example
\ref{ex:experiments} in the more general GPT case.  This provides another possible characterization for
classical-to-quantum channels, 
 in terms of guessing games with an inconclusive outcome,  where no ancilla is needed. 

For $\epsilon=0$, we obtain an ordering on the set of channels that was
treated also in \cite{buscemi2016degradable,buscemi2017comparison,chefles2009quantum,gour2018quantum}.
It was observed in \cite{buscemi2012comparison}  that in this case we may assume that the ensembles in (iii) consist of
separable states, see \cite[Theorem 2]{jencova2015comparison} for a proof close to the present setting. Equivalently, we may
restrict to separable states in (ii).
This also corresponds to the results of \cite{gour2018quantum}.

For similar results in the infinite dimensional case, see \cite{ganesan2020quantum}.

\subsection{Comparison by preprocessings}\label{sec:pre}

This time  the admissible spaces for $\Fe$ are $R_0A_1$ with fixed output system
$A_1$ and the morphisms in $\Fe(R_0A_1,S_0A_1)$ are restricted to preprocessings $\Phi\mapsto  \Phi\circ\Lambda$ for
some $\Lambda\in \Ce(S_0,R_0)$, so our problem has the form
 
\begin{center}
\begin{minipage}[c]{0.3\textwidth}
\centering
\includegraphics{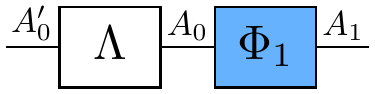}
\end{minipage}
\begin{minipage}[c]{0.01\textwidth}
 $\approx$
\end{minipage}
\begin{minipage}[c]{0.2\textwidth}
\centering
\includegraphics{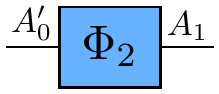}
\end{minipage}
 \end{center}

 Here $\Fe=\Ce_{\pre}*\Ce_{\post}$, where $\Ce_{\pre}=\Ce$ and $\Ce_{\post}=\{id_{A_1}\}$.
We have
\[
\delta_\Fe=\delta_{\pre}(\Phi_1\|\Phi_2):=\min_{\Lambda\in \Ce(A'_0,A_0)}\|\Phi_1\circ\Lambda-\Phi_2\|_\diamond.
\]
We will also consider the corresponding $\Fe$-distance, which  will be denoted by 
\[
\Delta_{\pre}:=\max\{\delta_{\pre}(\Phi_1\|\Phi_2),\delta_{\pre}(\Phi_2\|\Phi_1)\}.
\] 
A part of the following theorem was proved in \cite{jencova2014randomization}.
\begin{thm}\label{thm:pre} 
 Let $\Phi_1\in \Ce(A_0,A_1)$, $\Phi_2\in \Ce(A'_0,A_1)$ and let $\epsilon\ge 0$. The
following are equivalent.
\begin{enumerate}
\item[(i)] $\delta_{\pre}(\Phi_1\|\Phi_2)\le \epsilon$;

\item[(ii)] For any ancilla $R_0$ and  $\rho\in \states(R_0A_1)$, we have
\[
\|(id_{R_0}\otimes \Phi_2^{\mathsf{T}})(\rho)\|_{A'_0|R_0}^\diamond\le \|(id_{R_0}\otimes \Phi_1^{\mathsf{T}})(\rho)\|_{A_0|R_0}^\diamond+
\frac{\epsilon}2\|\rho\|_{A_1|R_0}^\diamond,
\]
\item[(iii)] For any ancilla $R_0$, any ensemble $\Ee$ on $R_0A_1$ and any fixed measurement $M$ on $A_1A_1$, we have
\[
P_{\succ}^{\Ce(R_0,A'_0),\Phi_2,\{M\}}(\Ee)\le P_{\succ}^{\Ce(R_0,A_0),\Phi_1,\{M\}}(\Ee)+\frac{\epsilon}2P_{\succ}(\Ee).
\]
\end{enumerate}

Moreover,  in (ii) and (iii), it is enough to use $R_0\simeq A'_0$.

\end{thm}

The guessing games in (iii) are depicted in the diagram
 
\begin{center}
\includegraphics{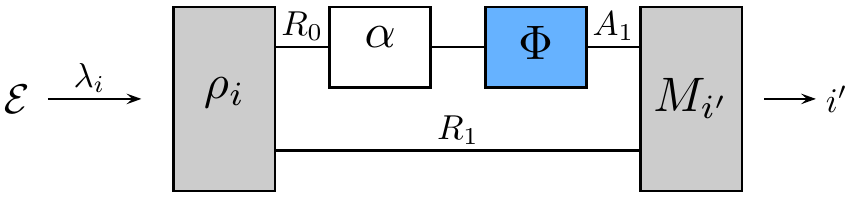}
\end{center}

Here the preprocessing $\alpha$ can be chosen freely and the measurement $M$ is fixed.

\begin{proof} The  equivalence (i) $\iff$ (ii) follow from Thm. \ref{thm:rand_chans} and 
\[
\<\rho,\Phi\circ\alpha\>= (\rho*C_\Phi)*C_\alpha=\<(id_{R_0}\otimes \Phi^{\mathsf{T}})(\rho),\alpha\>
\]
(see Remark \ref{rem:F_link}). In diagram:

\begin{center}
\begin{minipage}[c]{0.33\textwidth}
\includegraphics[align=c]{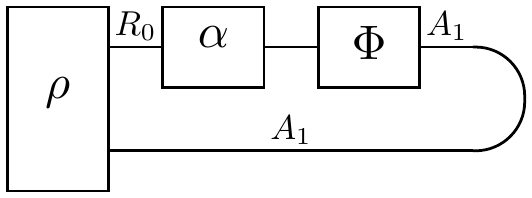}
\end{minipage}
\begin{minipage}[c]{0.03\textwidth}
\centering
${=}$
\end{minipage}
\begin{minipage}[c]{0.3\textwidth} 
{\includegraphics[align=c]{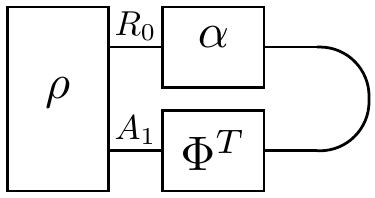}}
\end{minipage}
\end{center}

The implication (i) $\implies$ (iii) follows by Theorem \ref{thm:rand_chans_psuc}, with $\Me_{\post}=\{M\}$.
To finish the proof, we put $M=\mathsf{B}^{A_1}$ and observe 
 that  we have $\beta^M= id_{A_1}\in \Ce_{\post}$.

\end{proof}

Let us consider the more general situation when  $\Fe=\mathcal O*\{id_{A_1}\}$ where $\mathcal O(S_0,R_0)\subset
\Ce(S_0,R_0)$ is some suitable
subset.  Then
\[
\delta_\Fe(\Phi_1\|\Phi_2)=\inf_{\Lambda\in \mathcal O(A'_0,A_0)}\|\Phi_1\circ\Lambda-\Phi_2\|_\diamond.
\]
It can be seen as in  the above proof that we have $\delta_\Fe(\Phi_1\|\Phi_2)\le \epsilon$ if and only if 
for any ensemble $\Ee$ on $A'_0A_1$ and any fixed measurement $M$ on $A_1A_1$ we have
\[
P_{\succ}((\Phi_2\otimes id)(\Ee),M)\le \sup_{\Lambda\in \mathcal O}P_{\succ}((\Phi_1\circ\Lambda \otimes
id)(\Ee),M)+\frac{\epsilon}2P_{\succ}(\Ee).
\]
Let now $\Phi_1=id_{A_1}$, $\Phi_2=\Phi\in \Ce(A_0,A_1)$, then 
\[
\delta_\Fe(id\|\Phi)=\inf_{\Lambda\in \mathcal O(A_0,A_1)}\|\Lambda-\Phi\|_{\diamond},
\]
that is, the distance of $\Phi$ to the set $\mathcal O$. If $\mathcal O$ is the set of free channels in a resource theory for quantum channels, then 
this distance is a resource measure, \cite{liu2019resource}. The above considerations now give the
following operational characterization of this distance.

\begin{coro} $\inf_{\Lambda\in \mathcal F(A_0,A_1)}\|\Lambda-\Phi\|_{\diamond}\le \epsilon$ if and only if for any ensemble
$\Ee$ on $A_0A_1$ and any measurement $M$ on $A_1A_1$ we have
\[
P_{\succ}((\Phi\otimes id)(\Ee),M)\le \sup_{\Lambda\in \mathcal O} P_{\succ}((\Lambda\otimes
id)(\Ee),M)+\frac{\epsilon}2P_{\succ}(\Ee).
\]

\end{coro}

{

\subsubsection{$\Delta_{\pre}$ as the distance of ranges}

In this paragraph, we  obtain  a characterization of $\delta_{\pre}$ and the pseudo-distance $\Delta_{\pre}$ in terms of
the ranges of channels. Recall that the range of a channel $\Phi\in \Ce(A_0,A_1)$ is defined as
\[
\mathcal R(\Phi)=\Phi(\states(A_0)).
\]
Our first result in this direction is based on the following simple lemma. The proof is rather standard and is included
for the convenience of the reader.

\begin{lemma}\label{lemma:simple} Let $\sigma\in \states(AR)$ and let $Z\in \Be(R)$ be such that
$\sigma_R=\ptr_A[\sigma]=ZZ^*$. Then there is some channel
$\beta\in \Ce(R,A)$ such that  $\sigma=(\beta\otimes id)(|Z^{\mathsf{T}}\doublek Z^{\mathsf{T}}|)$. 

\end{lemma}

\begin{proof} Let $p=\mathrm{supp}(\sigma)$ and let $p_R=\mathrm{supp}(\sigma_R)$, then $p\le I\otimes p_R$ and we may
put 
\[
C:=(I\otimes U^*\sigma_R^{-1/2})\sigma(I\otimes \sigma^{-1/2}_RU)+d_A^{-1}I_A\otimes(I-U^*p_RU),
\]
where $U\in \Be(R)$ is a unitary such that $Z=\sigma^{1/2}_RU$ and the inverse is restricted to $p_R$. 
Since $C\ge 0$ and $\ptr_A[C]=I_R$, there is some $\beta\in \Ce(R,A)$ such that
$C=C_\beta$. We have
\[
\sigma=(I\otimes Z)C_\beta(I\otimes Z^*)=(\beta\otimes id)(|Z^{\mathsf{T}}\doublek Z^{\mathsf{T}}|).
\]

\end{proof}

\begin{coro}\label{coro:pre_ranges}
Let $\Phi_1\in \Ce(A_0,A_1)$, $\Phi_2\in \Ce(A'_0,A_1)$. Then 
\[
\delta_{\pre}(\Phi_1\|\Phi_2)=\sup_{\xi\in \states(A'_0R)}\inf_{\substack{\sigma\in \states(A_0R)\\
\sigma_{R}=\xi_{R}}}\|(\Phi_1\otimes id)(\sigma)-(\Phi_2\otimes id)(\xi)\|_1
\]
where $R\simeq A'_0$.
\end{coro}

\begin{proof} Let $\Lambda\in \Ce(A'_0,A_0)$ be such that $\|\Phi_1\circ\Lambda-\Phi_2\|_\diamond\le \epsilon$. Then for
any $\xi\in\states(A'_0R)$, we have
\[
\|(\Phi_1\circ\Lambda\otimes id)(\xi)-(\Phi_2\otimes id)(\xi)\|_1\le \|\Phi_1\circ\Lambda-\Phi_2\|_\diamond\le \epsilon.
\] 
Put $\sigma=(\Lambda\otimes id)(\xi)$, then we also have $\sigma_{R}=\xi_{R}$, so that the supremum  on the right hand side is upper
bounded by $\delta_{\pre}$. For the converse, let  $\rho\in \states(A'_0A_1)$.
By Lemma \ref{lemma:dualnorm_inf}, there is some $V\in \Be(A'_0)$, $\Tr[VV^*]=1$ and an element $G\in \Be_+(A'_0A_1)$
such that (recall that $\chi_V=V\cdot V^*$) 
\[
\rho=(\chi_V\otimes id)(G),\qquad \|\rho\|^\diamond_{A_1|A'_0}=\|G\|.
\] 
Using \eqref{eq:dualitylink} and \eqref{eq:dualityswp}, we have
\begin{align*}
\<\rho,\Phi_2\>&=\<(\chi_V\otimes id)(G),\Phi_2\>=\<G,\Phi_2\circ\chi_V\>=
\Tr[C_{\Phi_2\circ\chi_V}\tilde G]\\
&=\Tr[(\Phi_2\otimes id)(|V\doublek V|)\tilde G],
\end{align*}
where $\tilde G=\mathcal U^*_{A_1,A'_0}(G^{\mathsf{T}})$.
 Note that  $\xi:=|V\doublek V|\in \states(A'_0A'_0)$, with $\ptr_1[\xi]=V^{\mathsf{T}}(V^{\mathsf{T}})^*$. Assume that  there is some
$\sigma\in \states(A_0A'_0)$ with 
$ \sigma_{A'_0}=\ptr_1[\xi]$ and 
\[
\|(\Phi_2\otimes id)(\xi)-(\Phi_1\otimes id)( \sigma)\|_1\le \epsilon.
\]
By Lemma \ref{lemma:simple}, there is some channel $\beta\in \Ce(A'_0,A_0)$ such that 
$
\sigma=(\beta\otimes id)(|V\doublek V|).
$
We now have 
\begin{align*}
\<\rho,\Phi_2-\Phi_1\circ\beta\>&=\<G,(\Phi_2-\Phi_1\circ\beta)\circ\chi_V\>\\
&=
\Tr[((\Phi_2\otimes id)(\xi)-(\Phi_1\otimes id)(\sigma))\tilde G]\\
&\le \frac12\|(\Phi_2\otimes
id)(\xi)-(\Phi_1\otimes id)(\sigma)\|_1\|G\|\le \frac{\epsilon}2\|\rho\|^\diamond_{A_1|A'_0},
\end{align*}
here the  inequalities  follow from the fact that $\tilde G\in \Be_+(A_1A'_0)$, properties of the trace norm
$\|\cdot\|_1$  and $\|\tilde G\|=\|G\|=\|\rho\|^\diamond_{A_1|A'_0}$.
From Theorem \ref{thm:rand_chans} (ii), we obtain that $\delta_{\pre}(\Phi_1\|\Phi_2)\le \epsilon$.

\end{proof}

Using  the above corollary, we immediately obtain that for any $R$ with $d_R\ge d_{A'_0}$, 
$\delta_{\pre}(\Phi_1\|\Phi_2)=0$ is equivalent to the inclusion
\[
\mathcal R(\Phi_2\otimes id_{R})\subseteq \mathcal R(\Phi_1\otimes id_{R}).
\]
In the case of q-c channels, that is for measurements (POVMs), this result was proved in \cite{buscemi2005clean}, where also a counterexample was given, showing that inclusion of the ranges of the channels is not enough for existence
of even a positive preprocessing, so that tensoring with $id_R$ is necessary in general.

We next show that the pseudo-distance $\Delta_{\pre}$ can be expressed as a distance of  ranges. 
Recall that for two subsets $S,T$ of a metric space with metric $m$, the Hausdorff distance is defined by
\[
m_H(S,T)=\max\{ \sup_{s\in S}\inf_{t\in T}m(s,t), \sup_{t\in T}\inf_{s\in S}m(s,t)\}.
\] 
A natural choice for a metric on the set of states would be the trace distance 
\[
\|\sigma-\rho\|_1=\Tr|\sigma-\rho|,\qquad \sigma,\rho\in \states(A_1R).
\]
As it turns out, we will have to add a term for the  distance of the restrictions to $R$.
For  $\sigma_1,\sigma_2\in \states(R)$, let $p(\sigma_1,\sigma_2)$ denote the purified distance
\[
p(\sigma_1,\sigma_2)=\sqrt{1-F(\sigma_1,\sigma_2)^2}=\inf_{V_1V_1^*=\sigma_1,V_2V_2^*=\sigma_2}
\frac12 ||V^{\mathsf{T}}_1\doublek V^{\mathsf{T}}_1|- |V^{\mathsf{T}}_2\doublek V^{\mathsf{T}}_2|\|_1,
\]
where $F$ denotes the fidelity $F(\sigma_1,\sigma_2)= \|\sigma_1^{1/2}\sigma_2^{1/2}\|_1$.

\begin{coro}\label{coro:pre_dist}
For $\Phi_1\in\Ce(A_0,A_1)$ and $\Phi_2\in \Ce(A'_0,A_1)$,  we have
\[
\Delta_{\pre}(\Phi_1,\Phi_2)=m_H(\mathcal R(\Phi_1\otimes id_{R}),\mathcal R(\Phi_2\otimes id_{R})),
\]
where $d_R=\max\{d_{A_0},d_{A'_0}\}$ and $m_H$ is the Hausdorff distance with respect to the metric $m$ in $\states(A_1R)$, given as
\[
m(\xi,\sigma)=\|\xi-\sigma\|_1+2p(\xi_R,\sigma_R)
\]

\end{coro}

\begin{proof} From Corollary \ref{coro:pre_ranges},  we easily obtain that $m_H(\mathcal R(\Phi_1\otimes id_{R}),\mathcal
R(\Phi_2\otimes id_{R}))\le \Delta_{\pre}(\Phi_1,\Phi_2)$. For the converse, put 
\[
\epsilon:=m_H(\mathcal R(\Phi_1\otimes id_{R}),\mathcal
R(\Phi_2\otimes id_{R})).
\]
The idea of the proof is similar to the previous proof. Let 
$\rho\in \states(RA_1)$, $\alpha\in \Ce(R,A'_0)$ and let $V$ and $G$ be connected to $\rho$ as in the proof of Corollary
\ref{coro:pre_ranges}.  Let also $\xi:= (\alpha\otimes id)(|V\doublek V|)\in \states(A'_0R)$.
Then there is some $\sigma\in \states(A_0R)$ such that 
\[
m((\Phi_2\otimes id)(\xi),(\Phi_1\otimes id)(\sigma))\le \epsilon.
\]
Note that unlike the previous proof,  we may now have
$\xi_R\ne \sigma_R$. Let $W\in \Be(R)$ be such that $W^{\mathsf{T}}(W^{\mathsf{T}})^*=\sigma_R$ and 
\[
2p(\xi_R,\sigma_R)= \||V\doublek V|- |W\doublek W|\|_1
\]
(such $W$  always exists since $\|\cdot\|_1$ is unitarily invariant). By Lemma \ref{lemma:simple}, there is a channel $\gamma\in
\Ce(R,A_0)$ such that 
\[
\sigma=(\gamma\otimes id)(|W\doublek W|)=C_{\gamma\circ\chi_W}.
\]
We now have
\begin{align*} 
\<&\rho, \Phi_2\circ\alpha-\Phi_1\circ\gamma\>=\<G,(\Phi_2\circ\alpha-\Phi_1\circ\gamma)\circ\chi_V\>\\
&= \<G,\Phi_2\circ\alpha\circ\chi_V-\Phi_1\circ\gamma\circ\chi_W\>+\<G,\Phi_1\circ\gamma\circ(\chi_W-\chi_V)\>\\
&= \Tr[((\Phi_2\otimes id)(\xi)-(\Phi_1\otimes id)(\sigma))\tilde G]\\
&\qquad +\Tr[(\Phi_1\circ\gamma\otimes id)(|W\doublek
W|-|V\doublek V|)\tilde G]\\
&\le \frac12\|\rho\|^{\diamond}_{A_1|R} \bigl(  \|(\Phi_2\otimes id)(\xi)-(\Phi_1\otimes id)(\sigma)\|_1+\tilde
m_B(\xi_R,\sigma_R)\bigr)\le \frac{\epsilon}2\|\rho\|^{\diamond}_{A_1|R}.
\end{align*}
The last inequality implies that $\delta_{\pre}(\Phi_1\|\Phi_2)\le \epsilon$ and we similarly obtain that also 
$\delta_{\pre}(\Phi_2\|\Phi_1)\le \epsilon$.

\end{proof}

}

{

\subsection{Comparison of bipartite channels by LOCC superchannels}\label{sec:LOCC}

In this section, the objects in $\Fe$ are spaces of bipartite quantum channels $\Le(A_0B_0,A_1B_1)$ and 
 the morphisms are restricted to LOCC superchannels, that is, 
\[
\Fe=\Ce_2^{LOCC}:=\Ce_{LOCC}*\Ce_{LOCC}
\]
 where $\Ce_{LOCC}$ is the set of $A|B$ LOCC channels. To be more precise, in this case the admissible spaces
are of the form $R^AR^B=R^A_0R^B_0R^A_1R^B_1$ and the morphisms in $\Fe(R^AR^B, S^AS^B)$ have the form
 $\Lambda_{\pre}*\Lambda_{\post}$ with 
$\Lambda_{\pre}\in \Ce_{LOCC}(S^A_0|S^B_0,U^AR^A_0|R^B_0U^B)$, $\Lambda_{\post}\in \Ce_{LOCC}(U^AR^A_1|R^B_1U^B,S^A_1|S^B_1)$,
 so that also the ancilla consists of two parts $U=U^AU^B$. The simulation task becomes

\begin{center}
\begin{minipage}[c]{0.45\textwidth}
\centering
\includegraphics{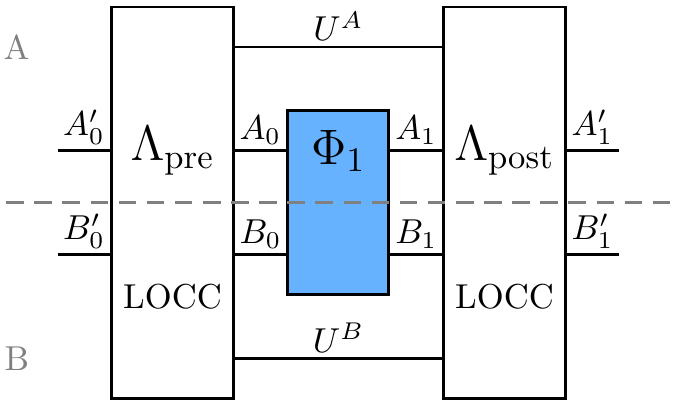}
\end{minipage}
\begin{minipage}[c]{0.01\textwidth}
 $\approx$
\end{minipage}
\begin{minipage}[c]{0.2\textwidth}
\centering
\includegraphics{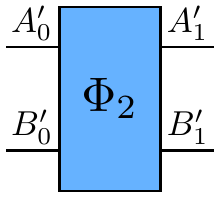}
\end{minipage}
 \end{center}

 In this case, 
\[
\delta_\Fe(\Phi_1\|\Phi_2)= \delta_{LOCC}(\Phi_1\|\Phi_2):=\inf_{\Lambda\in \Ce_2^{LOCC}}\|\Lambda(\Phi_1)-\Phi_2\|_\diamond
\]
 is the LOCC conversion distance $\Phi_1\to \Phi_2$. 
We will use the notation $\|\rho\|^{2-LOCC}_{S|R}:=\|\rho\|^\Fe_{S|R}$ for  $\rho\in \Be_+(R^AR^BS^AS^B)$.

\begin{thm}\label{thm:LOCC} Let $\Phi_1\in \Ce(A_0B_0,A_1B_1)$ and $\Phi_2\in \Ce(A_0'B_0',A_1'B_1')$, $\epsilon\ge 0$.
The following are equivalent.
\begin{enumerate}
\item[(i)] $\delta_{LOCC}(\Phi_1\|\Phi_2)\le \epsilon$;
\item[(ii)] for any spaces $R^AR^B$ and any $\rho\in \states(R^AR^B)$, we have
\[
\|\rho\otimes C_{\Phi_2}\|^{2-LOCC}_{R^AR^B|A'B'}\le \|\rho\otimes
C_{\Phi_1}\|^{2-LOCC}_{R^AR^B|AB}+\frac{\epsilon}2\|\rho\|^\diamond_{R^A_1R^B_1|R^A_0R^B_0}
\]
\item[(iii)] For any spaces $R^AR^B$ and any ensemble $\Ee$ on $R^AR^B$, we have
\[
P_{\succ}^{\Ce_{LOCC}(R^A_0R^B_0,\cdot A_0'B_0'\cdot) \Phi_2, \Me_{LOCC}}(\Ee)\le 
P_{\succ}^{\Ce_{LOCC}(R^A_0R^B_0,\cdot A_0B_0\cdot ),\Phi_1,\Me_{LOCC}}(\Ee)+\frac{\epsilon}2P_{\succ}(\Ee)
\]
where $\Me_{LOCC}$ is the set of LOCC measurements.
\end{enumerate}
Moreover, in (ii) and (iii) it is enough to take $R^A_0\simeq A'_0$, $R^A_1\simeq A'_1$, $R^B_0\simeq B'_0$,
$R^B_1\simeq B'_1$.

\end{thm}

The guessing games in (iii) have the form

\begin{center}
\includegraphics{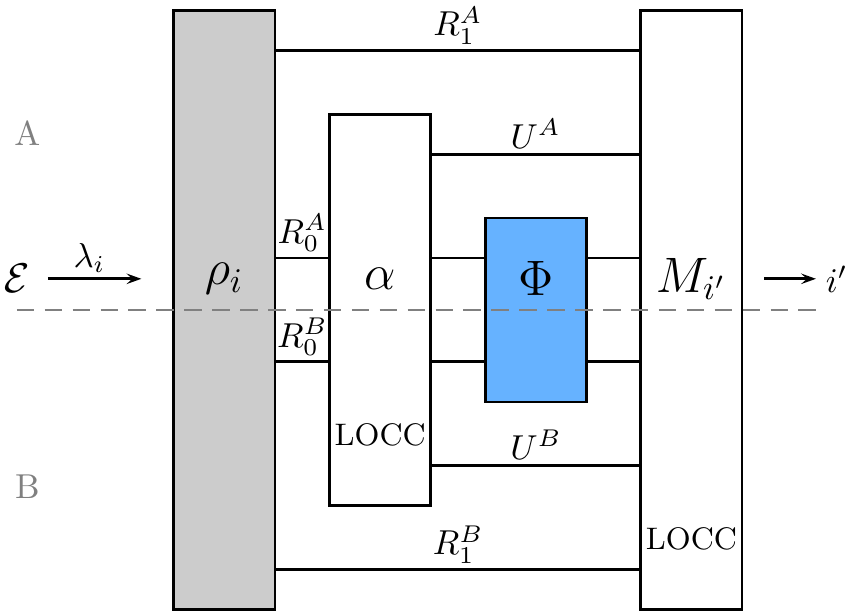}
\end{center}

Here $\alpha$ can be any  LOCC preprocessing channel and $M$ any LOCC measurement.

\begin{proof} The equivalence of (i) and (ii) is proved exactly as before from Theorem \ref{thm:rand_chans} and Remark
\ref{rem:F_link}. To prove the condition (iii), we invoke Theorem \ref{thm:rand_chans_psuc} with 
 $\Me_{\post}=\Me_{LOCC}$, this is obviously closed under $\Ce_{LOCC}$.
To prove the two additional conditions in Theorem \ref{thm:rand_chans_psuc}, observe that the group of generalized Pauli unitaries on $A'_1B'_1$ has the form
\[
\{U^{A'_1B'_1}_{x,y}=U^{A'_1}_x\otimes U^{B'_1}_y,\ x=1,\dots,d^2_{A'_1},\ y=1,\dots, d_{B'_1}^2\}
\]
and we have $\mathsf{B}^{A'_1B'_1}_{x,y}=\mathsf{B}^{A'_1}_x\otimes \mathsf{B}_y^{B'_1}\in \Me_{LOCC}(A_1'A_1'|B_1'B_1')$. 
It is now enough to show that for any measurement $M\in
\Me_{LOCC}(A_1'U^AA_1|B_1U^BB_1')$ with  outcomes labeled by $x,y$ we have $\beta^M\in
\Ce_{LOCC}(U^AA_1|B_1U^B,A_1'|B_1')$. Recall that $\beta^M$ is a channel associated with the instrument 
$\{\beta^M_{x,y}\}$, where the operations $\beta^M_{x,y}$ have the form

\begin{center}
\includegraphics{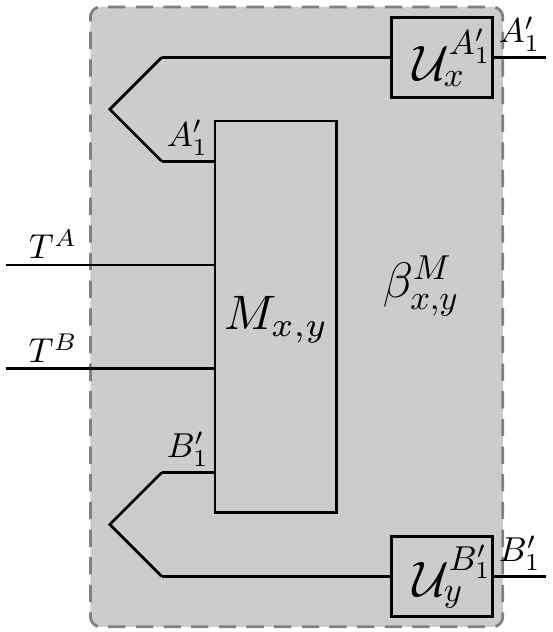}
\end{center}
here $T^A=U^AA_1$ and $T_B=B_1U^B$.
It is quite obvious that $\beta^M$ is LOCC if $M$ is.

\end{proof}

Note that we have a  similar situations for example for restricted PPT or SEP  superchannels,
 where  $\Ce_{\pre}=\Ce_{\post}=\Ce_{PPT}$ or $\Ce_{SEP}$. 

}

\subsection{Comparison of bipartite channels by partial superchannels}\label{sec:partial}

In this section, the objects of $\Fe$ are again spaces of bipartite channels, but this time 
 the morphisms are given by applying arbitrary  superchannels on the $A$ part, in diagram
\begin{center}
\begin{minipage}[c]{0.35\textwidth}
\centering
\includegraphics{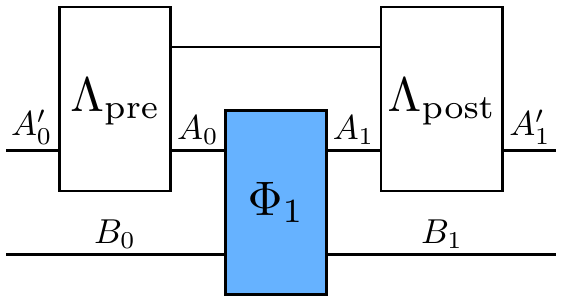}
\end{minipage}
\begin{minipage}[c]{0.01\textwidth}
\centering
\vspace{10mm}
${\approx}$
\end{minipage}
\begin{minipage}[c]{0.2\textwidth} 
\centering
\vspace{10mm}
{\includegraphics{pphi_bipart.t1.pdf}}
\end{minipage}
\end{center}

More precisely, the admissible spaces for $\Fe$ are of the form $RB$, where $R=R_0R_1$  are arbitrary and
$B=B_0B_1$ with $B_0$, $B_1$ fixed.  The morphisms $\Fe(RB, SB)$ in  $\Fe$ are
given by elements of $\Ce_2(R,S)$. Here we have $\Fe=\Ce_{\pre}*\Ce_{\post}$ with $\Ce_{\pre}=\Ce\otimes \{id_{B_0}\}$ and 
 $\Ce_{\post}=\Ce\otimes \{id_{B_1}\}$.

For $\Phi_1\in \Ce(A_0B_0,A_1B_1)$ and $\Phi_2(A'_0B_0,A_1'B_1)$ we denote
\[
\delta_{A|A'}(\Phi_1\|\Phi_2):=\delta_\Fe(\Phi_1\|\Phi_1)=\min_{\Theta\in \Ce_2(A,A')}\|(\Theta\otimes id_B)(\Phi_1)-\Phi_2\|_\diamond.
\]

\begin{thm}\label{thm:bipartite}
Let $\Phi_1\in \Ce(A_0B_0,A_1B_1)$, $\Phi_2\in \Ce(A'_0B_0,A'_1B_1)$ and let $\epsilon\ge 0$. The following are
equivalent.
\begin{enumerate}
\item[(i)] $\delta_{A|A'}(\Phi_1\|\Phi_2)\le \epsilon$;
\item[(ii)] For any spaces $R_0,R_1$ and $\rho\in \states(R_0B_0R_1B_1)$, we have
\[
\|\rho*C_{\Phi_2}\|^{2\diamond}_{R|A'}\le
\|\rho*C_{\Phi_1}\|^{2\diamond}_{R|A}+\frac{\epsilon}2\|\rho\|^{\diamond}_{R_1B_1|R_0B_0};
\]
\item[(iii)] For any spaces $R_0$, $R_1$, $k,l\in \mathbb N$, any ensemble $\Ee$ with $kl$ elements  on $R_0B_0R_1B_1$, any fixed measurement
$M\in \Me_l(B_1B_1)$, we have
\[
P_{\succ}^{\Ce(R_0,A'_0\cdot),\Phi_2,\Me_k\otimes M}(\Ee)\le P_{\succ}^{\Ce(R_0,A_0\cdot),\Phi_1,\Me_k\otimes
M}(\Ee)+\frac{\epsilon}2 P_{\succ}(\Ee).
\]
\end{enumerate}
Moreover, in (ii) and (iii) it is enough to put $R_0\simeq A_0'$, $R_1\simeq A_1'$. 

\end{thm}

The guessing games have  the form depicted in the diagram:
\begin{center}
\includegraphics{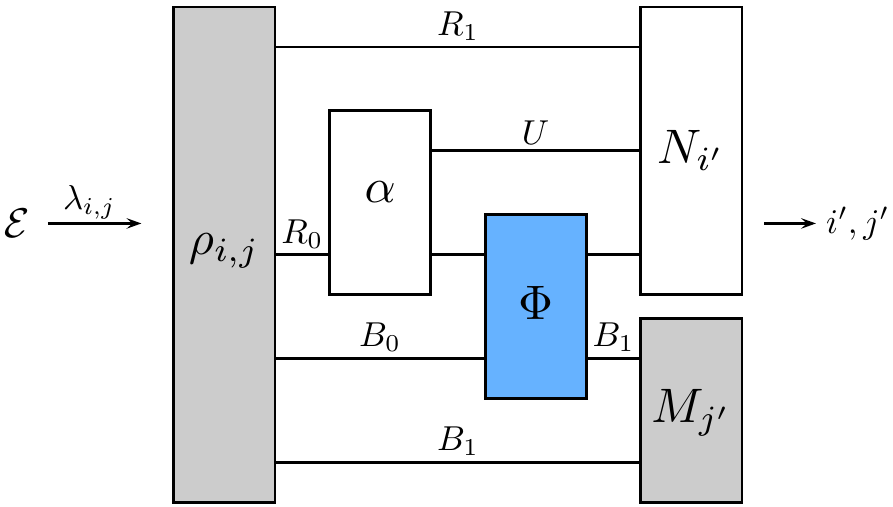}
\end{center}
Here $\Ee=\{\lambda_{i,j},\rho_{i,j}\}_{\substack{i=1,\dots,k\\ j=1,\dots,l}}$ is an and ensemble of states 
 on $R_0B_0R_1B_1$ and  $M\in \Me_l(B_1B_1)$ is  a fixed measurement. The preprocessing  $\alpha$ and the measurement $N$
with  $k$ outcomes can be chosen freely, with no restriction on the ancilla $U$.

\begin{proof} As before, the equivalence (i) $\iff$ (ii)  follows from Thm. \ref{thm:rand_chans} and the definition of the norm 
$\|\cdot\|^{2\diamond}$, taking into account 
 Remark \ref{rem:F_link}.

For the implication (i) $\implies$ (iii) we use Theorem \ref{thm:rand_chans_psuc} 
with  $\Me_{\post}=\Me_k\otimes M$, which is clearly closed under preprocessings from $\Ce_{\post}$. 
For the converse, let $l=d^2_{B_1}$, $k=d^2_{a_1'}$ and put the fixed measurement to $M_{y}=\mathsf{B}_y^{B_1}$.
 Then $\mathsf{B}^{A_1'B_1}=\mathsf{B}^{A'_1}\otimes \mathsf{B}^{B_1}\in \Me_k\otimes M=\Me_{\post}$.
Moreover, any measurement in $\Me_k\otimes M$ on $UA_1B_1A_1'B_1$ has the form $\tilde N=N\otimes
\mathsf{B}^{B_1}$
for some $N\in \Me_k(UA_1A_1')$.
The channel $\beta^{\tilde N}\in \Ce(UA_1B_1,A_1'B_1)$ clearly satisfies 
\[
\beta^{\tilde N}=\beta^{N}\otimes\beta^{\mathsf{B}^{B_1}}=\beta^N\otimes id_{B_1}\in \Ce_{\post}.
\]
The proof is finished by the second part of Theorem \ref{thm:rand_chans_psuc}.

\end{proof}

The case $\epsilon=0$ was also treated in \cite{gour2019comparison}. By a careful comparison of the results there to the
present setting, one can see that Theorem 7 in \cite{gour2019comparison} corresponds to our condition (ii), but
restricted to $\rho$ of the form 
\[
\rho=\sum_{x=1}^{d^2{B_0}}\sum_{y=1}^{d^2_{B_1}} |x\>\<x|\otimes E_y\otimes \Lambda_{y|x},
\]
where $|x\>\<x|$ is a normalized rank 1 basis of $B_0$, $E_y$ is an informationally complete POVM on $B_1$ and 
 $\Lambda_{y|x}\in \Le_+(R_0,R_1)$ is such that $\sum_y \Lambda_{y|x}$ is a channel for any $x$. The possibility of such
a restriction seems specific to the case $\epsilon=0$, as it is basically  a consequence of the fact that to check
equality of two channels it is enough to input some basis states and measure by an IC POVM.
Apart from some special cases (e.g. the one treated below), this cannot be extended to $\epsilon>0$ since in general
one needs entangled states to attain the diamond norm.

\subsection{Classical simulability of measurements}\label{sec:measurements}

As a further  application, we investigate the problem of classical simulability of measurements. In this problem, two sets of
measurements  $\Mbb=\{M^1,\dots,M^k\}$, $M^i\in \Me_l(C)$, $i=1,\dots,k$ and $\Nbb=\{N^1,\dots,N^m\}$,
 $N^y\in \Me_n(C)$, $y=1,\dots,m$ are given. 
We will say that $\Mbb$ can  simulate  $\Nbb$ if all elements in
$\Nbb$ can be obtained as convex combinations of postprocessings of elements in $\Mbb$,  \cite{guerini2017operational}. It can be seen  that we may exchange the order of convex combinations and postprocessings, and always
obtain the same notion of simulability. Our aim is to study an approximate version with respect to some suitable norm. In particular, we will
show that this problem can be put into the setting of Section \ref{sec:partial}: we represent $\Mbb$ and $\Nbb$
by  bipartite channels and  express the simulations as  applications of superchannels to one of the parts.

We will need the following type of guessing games. Let $\Ee=\{\lambda_x,\rho_x\}_{x=1}^n$ be an ensemble of states of 
$A$, but assume that only one fixed measurement $M\in \Me_l(A)$ can be performed and the true state has to be guessed using its outcome. Any
guessing procedure is described by conditional probabilities $\{p(x|j)\}$, giving the probability of
guessing $x\in \{1,\dots,n\}$ if $j\in\{1,\dots,l\}$ was measured. This defines a measurement 
 $N\in \Me_n(A)$ given as  
\[
N_x=\sum_j p(x|j)M_j,\quad x=1,\dots,n.
\]
Such measurement  is a postprocessing of $M$. The average probability of a correct guess using this procedure is
\[
P_{\succ}(\Ee,N)=P_{\succ}(\Ee,M,p):=\sum_{x,j} \lambda_x  p(x|j) \Tr[\rho_xM_j]
\] 
and the maximal success probability  is denoted by (cf. \cite{skrzypczyk2019robustness})
\[
P_{\succ}^Q(\Ee,M):=\sup_{\{p(x|j)\}} P_{\succ}(\Ee,M,p).
\]
 Any set of conditional probabilities can be identified with the classical-to-classical (c-c) channel in $\Ce(S,R)$, $d_S=l$,
$d_R=n$ determined by the Choi matrix
\[
C_p:=\sum_{x,j} p(x|j)|x\>\<x|\otimes|j\>\<j|.
\] 
This channel will be denoted by $p$. We then have $\Phi_N=p\circ \Phi_M$
for the  q-c channels given by the measurements $M$ and $N$. 
Moreover, for any  $\alpha\in \Ce(S,R)$, there are conditional probabilities $\{p(x|j):=\<x|\alpha(|j\>\<j|)|x\>\}$, such
that
\begin{equation}\label{eq:psucc_Q_channel}
\<(\Phi_M\otimes id)(\rho_\Ee),\alpha\>= \<\rho_\Ee,\alpha\circ \Phi_M\>= \<\rho_\Ee,p\circ \Phi_M\>= P_{\succ}(\Ee,M,p),
\end{equation}
This proves the following result.

\begin{lemma}\label{lemma:psuc_Q} For any ensemble $\Ee$ on $A$ and $M\in \Me_{d_S}(A)$,
\[
P_{\succ}^Q(\Ee,M)=\|(\Phi_M\otimes id)(\rho_\Ee)\|^\diamond_{S|A}= P_{\succ}(\Phi_M(\Ee)).
\]

\end{lemma}

Let now $\Mbb$ and $\Nbb$ be sets of measurements as above. By definition, $\Mbb$ can simulate $\Nbb$ if for each  
$N^y\in \Nbb$ there are probabilities $q(i|y)$ such that 
\[
\Phi_{N^y}=\sum_{i} q(i|y) p_{i,y}\circ\Phi_{M^i}
\]
for some c-c channels $p_{i,y}$ determined by sets of conditional probabilities $\{p_{i,y}(x|j),\ x=1,\dots,n,\ j=1,\dots, l\}$,
$i=1,\dots,k$, $ y=1,\dots, m$. That is,  
\begin{equation}\label{eq:simul}
N^y_x=\sum_{i,j}p(i,x|j,y) M^i_j,\qquad x=1,\dots,n,\ y=1,\dots,m,
\end{equation}
where $p(i,x|j,y):=q(i|y)p_{i,y}(x|j)$ are  conditional
probabilities.   Let $\Phi_\Mbb$ be a channel in $\Ce(A_0C,A_1)$, $d_{A_0}=k$, $d_{A_1}=l$ with  the Choi matrix
\[
C_\Mbb:= \sum_{i,j} |j\>\<j|\otimes |i\>\<i|\otimes (M^i_j)^{\mathsf{T}}.
\]
Note that $\Phi_\Mbb$ is a bipartite channel as in the setting of Theorem \ref{thm:bipartite}, with $B_0=C$ and $B_1=1$,
moreover, the first input and the output of $\Phi_\Mbb$ is classical.
Similarly, $\Nbb$  is represented by the  channel $\Phi_\Nbb\in \Ce(A'_0C,A'_1)$, with $d_{A'_0}=m$, $d_{A'_1}=n$. It
can be easily checked that  \eqref{eq:simul} can be expressed as
\[
C_p*C_\Mbb=C_{\Nbb},
\] 
where  $p\in \Ce(A_0'A_1,A_0A_1')$ is the c-c channel given by $\{p(i,x|j,y)\}$. Using \eqref{eq:comb}, we ca see that
$p$ is a superchannel if and only if there are conditional probabilities $\{q(i|y)\}$ such that 
\[
\sum_x p(x,i|y,j)=q(i|y),\qquad \forall i,y,\ \forall j.
\]
By putting $p_{i,y}(x|j):=q(i|y)^{-1}p(x,i|y,j)$ if $q(i|y)> 0$ and choosing any conditional probabilities for
$p_{i,y}$ otherwise, we obtain the following result.

\begin{lemma}\label{lemma:cc_combs} A c-c channel $p\in \Ce(A'_0A_1,A_0A'_1)$ is a superchannel if and only if  there are conditional
probabilities $\{p_{i,y}(x|j)\}$ and $\{q(i|y)\}$, with $i=1,\dots, d_{A_0}$, $j=1,\dots,d_{A_1}$, $x=1,\dots,d_{A'_1}$, 
$y=1,\dots,d_{A'_0}$ such that
\[
p(x,i|y,j)=q(i|y)p_{i,y}(x|j).
\]

\end{lemma}

We obtain that $\Nbb$ is simulable by $\Mbb$ if and only if   $\Phi_{\Nbb}=p(\Phi_\Mbb)$ for some c-c superchannel $p$, 
in diagram
\begin{center}
\begin{minipage}[c]{0.2\textwidth} 
\centering
\vspace{10mm}
\includegraphics{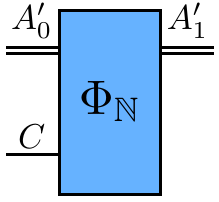}
\end{minipage}
\begin{minipage}[c]{0.01\textwidth}
\centering
\vspace{10mm}
$=$
\end{minipage}
\begin{minipage}[c]{0.35\textwidth}
\centering
\includegraphics{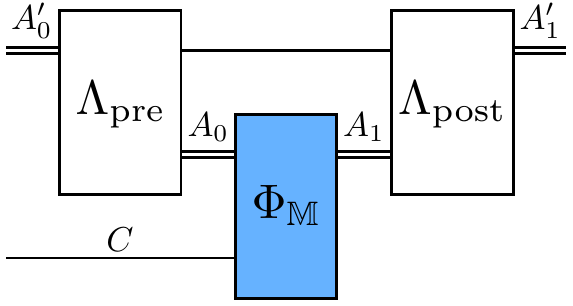}
\end{minipage}

\end{center}
here the double lines denote classical inputs and outputs.

We now introduce the following notion of approximate simulability: for $\epsilon\ge 0$, we say that $\Nbb$ is $\epsilon$-simulable by 
$\Mbb$ if $\Mbb$ can simulate some set of measurements $\Nbb'=\{(M')^1,\dots, (M')^m\}\subset \Me_n(C)$ such that 
\[
\|\Phi_\Nbb-\Phi_{\Nbb'}\|_\diamond\le \epsilon.
\]

The next result shows that approximate simulability is expressed by the conversion distance 
$\delta_{A|A'}(\Phi_\Mbb\|\Phi_\Nbb)$.
\begin{prop}\label{prop:simul_defic} Let $\Mbb=\{M^1,\dots,M^k\}\subset \Me_l(C)$, $\Nbb=\{N^1,\dots,N^m\}\subset
\Me_n(C)$, $\epsilon\ge 0$. Let $A_0,A_1$, $A'_0,A'_1$ be systems such that $d_{A_0}=k$, $d_{A_1}=l$, $d_{A'_0}=m$,
$d_{A'_1}=n$. Then $\Nbb$ is $\epsilon$-simulable by $\Mbb$ if and only if 
\[
\delta_{A|A'}(\Phi_\Mbb\|\Phi_\Nbb)\le \epsilon.
\]
\end{prop}

\begin{proof} Assume that $\Nbb$ is $\epsilon$-simulable by $\Mbb$. As shown above, there is some c-c  superchannel
$p\in \Ce_2(A,A')$ such that $\|p(\Phi_{\Mbb})-\Phi_{\Nbb}\|_\diamond\le \epsilon$ and hence  
$\delta_{A|A'} (\Phi_\Mbb\|\Phi_{\Nbb})\le \epsilon$. For the converse, for any system $D$, let $\Delta_D\in \Ce(D,D)$
denote the channel that maps any  $X\in \Be(D)$ to its diagonal elements in the  basis $|i\>_D$:
\[
\Delta_D(X)=\sum_{i=1}^{d_D} \<i|X|i\>|i\>\<i|.
\] 
For any $\Theta\in \Ce_2(A,A')$, we have
\[
\|\Theta(\Phi_\Mbb)-\Phi_\Nbb\|_\diamond\ge \|\Delta_{A'_1}\circ (\Theta(\Phi_\Mbb)-\Phi_\Nbb)\circ\Delta_{A'_0}\|_\diamond=
\|\Delta_{A'_1}\circ\Theta(\Phi_\Mbb)\circ\Delta_{A'_0}-\Phi_\Nbb \|_\diamond,
\]
so that the minimum in $\delta_{A|A'}(\Phi_\Mbb\|\Phi_\Nbb)$ is attained at some  c-c superchannel $p$ and we have seen 
that  $p(\Phi_\Mbb)=\Phi_{\Nbb'}$ for some set $\Nbb'$ of measurements that are
simulated by  $\Mbb$. It follows that $\Nbb$ is
 $\epsilon$-simulable by $\Mbb$.

\end{proof}

We are ready to apply the results of Section \ref{sec:partial} and prove that 
 $\epsilon$-simulability can be characterized by guessing games. Note that  in this case, it is
 enough to use  ensembles on the system $C$ (so $R_0=R_1=1$ in Theorem \ref{thm:bipartite} (iii)). For $\Mbb$ 
 consisting of a
single element and $\epsilon=0$, this result was proved in \cite{skrzypczyk2019robustness}.

\begin{coro}\label{coro:sc_simul} Let $\Mbb$ and $\Nbb$ be sets of measurements as above, $\epsilon\ge 0$. Then
 $\Nbb$ is $\epsilon$-simulable by $\Mbb$
if and only if  for any ensemble $\Ee$ on $C$, 
\[
\max_{1\le y\le m} P_{\succ}^Q(\Ee,N^y)\le \max_{1\le i\le k} P_{\succ}^Q(\Ee,M^i) + \frac{\epsilon}2P_{\succ}(\Ee).
\]

\end{coro}

\begin{proof} We start by expressing the success probabilities of  part (iii) of Theorem \ref{thm:bipartite} for $R_0=R_1=1$. 
 Let $\Ee=\{\lambda_a,\rho_a\}$ be an ensemble on $C$. Since $B_0=C$ and $B_1=1$,   the guessing games with
$\Phi_\Mbb$  can be represented as in the diagram 

\begin{center}
\includegraphics{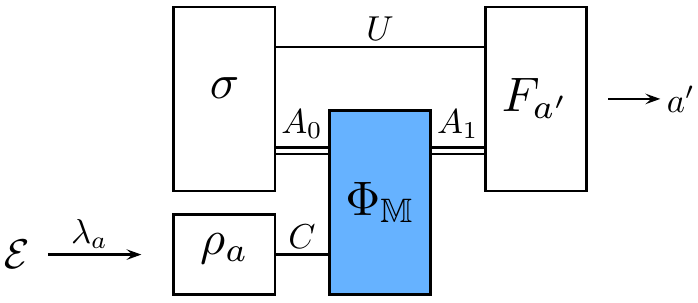}
\end{center}
Here, as a preprocessing, we pick a quantum-classical state $\sigma\in \states(UA_0)$ seen as a channel in $\Ce(1,SA_0)$ and we also pick a measurement $F$ on $UA_1$.
The success probability is then by \eqref{eq:Psucc_e_m} 
\[
P_{\succ}(\sigma\otimes \Ee, \Phi_\Mbb^*(F))=\< \rho_\Ee, \Theta(\Phi_\Mbb)\>,
\]
with  $\Theta:=\sigma*\Phi_F\in \Ce_2(A,Q)$, where  $Q=Q_0Q_1$, $Q_0=1$, $Q_1$ is the output system of 
$\Phi_F$. 
Since $\Theta$ is obviously a c-c superchannel, by Lemma \ref{lemma:cc_combs} there are conditional probabilities 
$p_i(a|j)$ and probabilities $q(i)$ such that $\Theta=p$ with $p(a,i|j)=q(i)p_i(a|j)$.  We obtain
\begin{align*}
\< \rho_\Ee, \Theta(\Phi_\Mbb)\>&= \rho_\Ee*C_p*C_\Mbb= \sum_i q(i) \sum_{j,a}
\lambda_a p_i(a|j)\Tr[\rho_aM^i_j]\\
&= \sum_i q(i) P_{\succ}(\Ee, M^i,p_i).
\end{align*}
Since any c-c superchannel in $\Ce_2(A,Q)$ consists of a preprocessing and postprocessing of the above form, we see that 
\begin{align*}
P_{\succ}^{\Ce(1,\cdot A_0),\Phi_\Mbb,\Me}(\Ee)&=\sup_{\sigma, F} P_{\succ}(\sigma\otimes \Ee,\Phi_\Mbb^*(F))=
\sup_{p\in \Ce_2(A,Q)} \<\rho_\Ee,p(\Phi_\Mbb)\>\\
&=\sup_{q,\{p_i\}} \sum_i q(i) P_{\succ}(\Ee, M^i,p_i)=\\
&=\sup_q \sum_i q(i) P_{\succ}^Q(\Ee, M^i)=\max_{1\le i\le m}P_{\succ}^Q(\Ee, M^i),
\end{align*}
where the second supremum is taken over all c-c superchannels. Since we have a similar equality for $\Phi_\Nbb$, we
obtain the 'only if' part.

For the converse, let $\rho\in \states(A'_0CA_1)$. It is easy to see from the shape of the Choi matrix $C_\Nbb$ that there are
probabilities $\lambda(y)$, conditional probabilities $\mu(x|y)$ and states $\rho^y_x\in \states(C)$ such that 
\[
\<\rho,\Phi_\Nbb\>=\rho* C_\Nbb=\sum_y \lambda(y) \sum_x\mu(x|y) \Tr[\rho^y_xN^y_x]=\sum_y\lambda(y) P_{\succ}(\Ee_y,N^y),
\]
here $\Ee_y=\{\mu(x|y), \rho^y_x\}_{x=1}^n$. For each $y$, $P_{\succ}(\Ee_y,N^y)\le P_{\succ}^Q(\Ee_y,N^y)$, so that by the
assumption, there is some
$1\le i_y\le k$ and conditional probabilities $p_{y}(x|j)$ such that 
\[
P_{\succ}(\Ee_y,N^y)\le P_{\succ}(\Ee_y, M^{i_y},p_{y})+\frac{\epsilon}2P_{\succ}(\Ee_y).
\]
Put $q(i|y)=\delta_{i,i_y}$, then $q(i|y)$ are conditional probabilities. Put
$p(i,x|y,j)=q(i|y)p_{y}(x|j)$, then $p$ is a c-c superchannel in $\Ce_2(A,A')$. It can be easily computed that 
\begin{align*}
\<\rho,p(\Phi_\Mbb)\>&=\rho*C_p*C_\Mbb=\sum_{i,j,x,y} \lambda(y)\mu(x|y)\Tr[\rho^y_xM^i_j]p(i,x|y,j)\\
&= \sum_{y,x,j} \lambda(y)\mu(x|y)\Tr[\rho^y_xM^{i_y}_j]p_y(x|j)=\sum_y \lambda(y) P_{\succ}(\Ee_y,M^{i_y},p_y)
\end{align*}
It follows that 
\[
\<\rho,\Phi_\Nbb\>\le \<\rho,p(\Phi_\Mbb)\>+\frac{\epsilon}2\sum_y\lambda(y)P_{\succ}(\Ee_y).
\]
Let now $F^y\in \Me_n(C)$ be such that $P_{\succ}(\Ee_y)=P_{\succ}(\Ee_y,F^y)$ and let $\mathbb F=\{F^1,\dots, F^m\}$. 
As we have seen before,
\[
\sum_y\lambda(y)P_{\succ}(\Ee_y)=\sum_y\lambda(y) P_{\succ}(\Ee_y,F^y)=\<\rho,\Phi_{\mathbb F}\>\le
\|\rho\|^\diamond_{A'_1|A'_0C}.
\]
By Theorem \ref{thm:rand_chans} (ii), this finishes the proof.

\end{proof}

\section{Concluding remarks}

We have introduced a general framework for comparison of channels, in quantum 
information theory as well as in the broader setting of GPT. 
The framework is based on
 the category $\mathsf{BS}$, which is a special category of ordered (finite dimensional) vector spaces,  modelled on the
set of channels.  In this setting, we 
defined a notion of an $\Fe$-conversion distance $\delta_\Fe$ with respect to a convex subcategory $\Fe$ and proved a general
result giving an operational  characterization of this distance.  
This result was then applied to  quantum channels, where we proved that the $\Fe$-conversion distance can be characterized 
by a set of  modified conditional min-entropies. In the setting of quantum resource theories of processes, these
quantities form a complete set of resource monotones. Under some conditions of the subcategory $\Fe$, the modified
conditional min-entropies can be obtained by a conic program.

In the case when elements of  $\Fe$  can be characterized as
concatenations $\Theta_{\pre}*\Theta_{\post}$ with $\Theta_{\pre}\in \Ce_{\pre}$ and $\Theta_{\post}\in \Ce_{\post}$ for some
suitable sets of channels, we also characterized the $\Fe$-conversion distance in terms of success probabilities in some
guessing games. We discussed several choices of such subcategories:
postprocessings, preprocessings, LOCC and partial superchannels on bipartite channels. We also noted that our results 
 hold for other classes of restricted superchannels, such as PPT or SEP. As another application, we studied the problem
of approximate classical simulability for sets of measurements. 

The advantage of the general formulation in the GPT setting of Section \ref{sec:gpt} is that our results  can be applied not only to pairs of 
channels but also to more specialized networks. In particular, one can study the problem of converting $m$ copies of a channel to
$n$ copies  of another, using parallel or sequential schemes (see also \cite{liu2020operational}). If  the $m$ copies of a channel, or more generally an
$m$-tuple of channels are used in parallel, we can treat their tensor product simply as a channel and some variants of the present theory
may be applied. If we use a sequential scheme with respect to some fixed ordering, the tensor product is a special case
of an $m$-comb \cite{chiribella2009theoretical}. Similarly to 2-combs, spaces of $m$-combs are also objects of the category $\mathsf{BS}$, so convertibility in this
setting can be treated using  Theorem \ref{thm:rand}. On the other hand, if there is no fixed ordering in the use of the
channels, then we may see the tensor product of the channels as a product  element in an object obtained from a symmetric
monoidal structure in $\mathsf{BS}$. Note also that the choice of the subcategory $\Fe$ not only restricts the allowed
transformations (by the choice of morphisms) but also determines the distance measure in $\delta_\Fe$ (by the choice of
the objects).

This work concerns only the one shot situation, when the channels in question are used only once. To go beyond the one
shot setting in the general framework, we need to discuss possible symmetric  monoidal structures (tensor products) in
$\mathsf{BS}$, their properties and the corresponding behaviour of the related norms. Another important
direction is an extension to infinite dimensions. The present framework strongly depends on the finite dimensional
setting, but some corresponding  results for  post- and preprocessings for quantum  channels  on semifinite von Neumann
algebras were proved in \cite{ganesan2020quantum}.

\section*{Acknowledgements}

I am indebted to  Martin Pl\'avala for the idea of introducing  the category that is here called
$\mathsf{BS}$ and sharing his notes, as well as, together with Gejza Jen\v{c}a, for discussions on its definition and
properties. I am also grateful to anonymous referees whose comments helped me greatly to improve the
presentation and results of the paper. 
Part of the results (Section 3.5) were presented at the Algebraic and Statistical ways into Quantum Resource Theories workshop at BIRS
 in July 2019. The research was supported by grants VEGA 2/0142/20 and APVV-16-0073.

%\bibliography{PROJECT_comparison}{}
%\bibliographystyle{hieeetr}

%\end{document}

\end{document}